\documentclass[twoside,leqno,twocolumn]{article}

\newif\ifdebug
\debugtrue

\newif\ifwarn
\warnfalse

\newif\ifsdm
\sdmfalse

\usepackage{ltexpprt}
\usepackage{afterpage}
\usepackage[utf8]{inputenc}
\usepackage{dsfont,amsmath,amssymb,amsfonts,upgreek,bbold,marvosym} 
\usepackage[dvipsnames]{xcolor}
\usepackage{graphicx}
\usepackage{verbatim}
\usepackage{url,hyperref}
\usepackage{float}
\usepackage{mathtools}
\usepackage{mathptmx}
\usepackage{tikz}
\usetikzlibrary{arrows,matrix}
\usepackage{etoolbox}
\usepackage[perpage,para,symbol*]{footmisc}
\usepackage{xspace}
\usepackage{ifthen}
\usepackage{enumitem}
\usepackage{algorithm,algorithmic}
\usepackage{subcaption}
\graphicspath{{./fig/}}

\ifdebug
    \definecolor{myblue}{rgb}{0.0, 0.0, 1.0}
    \definecolor{mypink}{rgb}{1.0, 0.01, 0.24}
\else
    \definecolor{myblue}{rgb}{0.0, 0.0, 0.0}
\fi

\newcommand{\bull}{\ensuremath{\boldsymbol{\triangleright}}\xspace}

\newcommand{\onebb}[0]{\mathds{1}}

\newcommand{\diag}[0]{\ensuremath\operatorname{diag}}
\newcommand{\bigoh}[0]{\ensuremath{O}}

\newcommand{\trans}[0]{\ensuremath\intercal}

\renewcommand{\emptyset}[0]{\varnothing}
\renewcommand{\alpha}[0]{\upalpha}
\renewcommand{\beta}[0]{\upbeta}
\renewcommand{\lambda}[0]{\uplambda}
\renewcommand{\sigma}[0]{\upsigma}
\renewcommand{\iff}[0]{\Leftrightarrow}

\DeclareMathOperator*{\argmin}{arg\,min}

\newcommand{\Wnew}[0]{\ensuremath{\widetilde{W}}}
\newcommand{\Wkc}[0]{\ensuremath{W^{KC}}}
\newcommand{\Wkcnew}[0]{\ensuremath{\widetilde{W}^{KC}}}
\newcommand{\Wbinnew}[0]{\ensuremath{\widetilde{W}_{01}}}
\newcommand{\xnew}[0]{\ensuremath{\widetilde{x}}}
\newcommand{\pinew}[0]{\ensuremath{\widetilde{\pi}}}

\newcommand{\expect}[1]{\ensuremath{\operatorname{\mathbb{E}}\left[#1\right]}}
\newcommand{\edgeweight}[1]{\ensuremath{\theta_{#1}}}
\newcommand{\gap}[0]{\ensuremath{\operatorname{gap}}}
\newcommand{\acv}[0]{\ensuremath{\langle \pi, x(0) \rangle}\xspace}
\newcommand{\acvnewx}[0]{\ensuremath{\langle \pi, \xnew(0) \rangle}\xspace}
\newcommand{\acvnewxpi}[0]{\ensuremath{\langle \pinew, \xnew(0) \rangle}\xspace}
\newcommand{\acvs}[0]{\ensuremath{\langle \pi, x \rangle}\xspace}
\newcommand{\acvnewxs}[0]{\ensuremath{\langle \pi, \xnew \rangle}\xspace}
\newcommand{\acvnewxpis}[0]{\ensuremath{\langle \pinew, \xnew \rangle}\xspace}

\newcommand{\ssp}[0]{\ensuremath{\operatorname{SSP}}\xspace}
\newcommand{\knapsack}[0]{\ensuremath{\operatorname{KNAPSACK}}\xspace}
\newcommand{\kssp}[0]{\ensuremath{\operatorname{kSSP_{01}}}\xspace}
\newcommand{\diver}[0]{\ensuremath{\operatorname{DIVER}}\xspace}

\patchcmd{\bordermatrix}{8.75}{4.75}{}{}
\patchcmd{\bordermatrix}{\left(}{\left[}{}{}
\patchcmd{\bordermatrix}{\right)}{\right]}{}{}

\makeatletter
\newcommand*\rel@kern[1]{\kern#1\dimexpr\macc@kerna}
\newcommand*\widebar[1]{%
  \begingroup
  \def\mathaccent##1##2{%
    \rel@kern{0.8}%
    \overline{\rel@kern{-0.8}\macc@nucleus\rel@kern{0.2}}%
    \rel@kern{-0.2}%
  }%
  \macc@depth\@ne
  \let\math@bgroup\@empty \let\math@egroup\macc@set@skewchar
  \mathsurround\z@ \frozen@everymath{\mathgroup\macc@group\relax}%
  \macc@set@skewchar\relax
  \let\mathaccentV\macc@nested@a
  \macc@nested@a\relax111{#1}%
  \endgroup
}
\makeatother

\newcommand{\citex}[2]{\cite[#2]{#1}}

\newtheorem{definition}{Definition}
\newtheorem{assumption}{Assumption}

\newcommand{\vsa}{\vspace*{-0.1cm}}

\newcommand{\vsbb}{\vspace*{-0.3cm}}
\newcommand{\vsc}{\vspace*{-0.4cm}}

\newcommand*{\qedbull}{\hfill\ensuremath{\square}}

\ifwarn
    \usepackage{background}
    \SetBgContents{{\color{gray}\textbf{
        Unpublished material. Do not re-distribute~\copyright~2017 Victor Amelkin
    }}}
    \SetBgPosition{-0.65in,-4.5in}
    \SetBgOpacity{1.0}
    \SetBgAngle{90.0}
    \SetBgScale{1}
\fi

\title{ \Large Disabling External Influence in Social Networks via
    Edge Recommendation}

\author{%
Victor Amelkin~\thanks{University of California, Santa Barbara (\texttt{victor@cs.ucsb.edu})}~~\footnotemark[3]
\and
Ambuj K. Singh~\thanks{University of California, Santa Barbara (\texttt{ambuj@cs.ucsb.edu})}~~\thanks{This work was supported by the U.~S.~Army Research
Laboratory and the U.~S.~Army Research Office under grant number W911NF-15-1-0577.}%
}

\date{}

\begin{document}

\maketitle

\ifsdm
\fancyfoot[R]{\footnotesize{\textbf{%
Copyright \textcopyright\ 2018 by SIAM\\Unauthorized reproduction of this article is prohibited%
}}}
\fi

\begin{abstract} \small\baselineskip=9pt
    Existing socio-psychological studies suggest that users of a social network
    form their opinions relying on the opinions of their neighbors. According to
    DeGroot opinion formation model, one value of particular importance
    is the asymptotic consensus value $\langle \pi, x \rangle$---the sum of user
    opinions $x_i$ weighted by the users' eigenvector centralities $\pi_i$. This
    value plays the role of an attractor for the opinions in the network and is
    a lucrative target for external influence. However, since any potentially
    malicious control of the opinion distribution in a social network is clearly
    undesirable, it is important to design methods to prevent the external attempts
    to strategically change the asymptotic consensus value.
    In this work, we assume that the adversary wants to maximize the asymptotic
    consensus value by altering the opinions of some users in a network; we, then,
    state \diver---an NP-hard problem of disabling such external influence attempts by
    strategically adding a limited number of edges to the network. Relying on the
    theory of Markov chains, we provide perturbation analysis that shows how
    eigenvector centrality and, hence, \diver's objective function change in
    response to an edge's addition to the network. The latter leads to the design of a pseudo-linear-time heuristic for \diver, whose
    computation relies on efficient estimation of mean first passage times in a
    Markov chain. We confirm our theoretical findings in experiments.
\end{abstract}

\ifdebug
    \setcounter{page}{1}
    \pagenumbering{arabic}
    \afterpage{\cfoot{\thepage}}
\fi

\section{Introduction}

Online social network play an important role in today's life, which is, to a great
extent, due to the fact that, in the absence of the objective means for opinion
evaluation, people tend to evaluate their opinions by comparison with the opinions
of others~\cite{festinger1954theory}. Thus, social networks impact the opinion
formation process in the society. Clearly, the society would benefit from this
process' being natural and fair, with good ideas spreading and bad ideas disappearing.
However, viral marketing experts may be interested in affecting the opinion
formation process, having the goal of driving the opinion distribution to a certain
business-imposed objective. One popular way of affecting (or controlling) the opinion
formation process is influence maximization~\cite{domingos2001mining}, whose central
idea is to affect the opinions of a limited number of users in the network with the
goal of maximizing the subsequent spread of ``right'' opinions from these users
throughout the network, or, more generally, shifting the opinion distribution towards
a desired state. Naturally, the society would benefit from having a mechanism that
would prevent such potentially malicious interventions into the opinion formation
process in social networks. Our work is dedicated to the design of one such
mechanism---an edge recommendation algorithm that disables the effect of the attempts
to control the opinion distribution in an online social network through user influence.

In this work, we assume that the user opinions $x(t)$ are formed in the network following
the well-established DeGroot(-Abelson) opinion formation model~\cite{degroot1974reaching,abelson1964mathematical}
\begin{align*}
    x(t+1) = W x(t),\ x(t) \in [0, 1]^n,\ W \in [0, 1]^{n \times n},\ W\onebb = \onebb,
\end{align*}
where $t$ is time, and $W$ is a row-stochastic interpersonal appraisal matrix---playing
the role of the adjacency
matrix of a directed social network---whose element $w_{ij}$ measures the relative extent
to which user $i$ values the opinion of user $j$ (see Fig.~\ref{fig:diver-by-example}a).
According to this model, users form
their opinions via weighted averaging of their own opinions with those of their
neighbors in the network. The model's rationale---buttressed by social comparison theory~\cite{festinger1954theory},
cognitive dissonance theory~\cite{festinger1962theory}, and balance 
theory~\cite{heider1946attitudes,cartwright1956structural}---is that people act to achieve
balance with other group members or, alternatively, to relieve psychological discomfort
from disagreement with others. Its well-known that, in a long term, in a ``well-connected'' social network, the opinions of all users approach the same \emph{asymptotic consensus value}
$$
    \lim\limits_{t \to \infty}{x_i(t)} = \acv = \pi^\trans x(0),
$$
being a sum of the initial opinions $x_i(0)$ of all the users, weighted by the users'
eigenvector centralities $\pi_i$. While in real-world situations, people do not always
agree upon the same opinion---in contrast to how it is prescribed by DeGroot model---$\acv$
still can be viewed as the value to which the opinions of all the network users are
attracted, which makes this value a lucrative target for influence.

\begin{figure*}[ht]
    \centering
    \includegraphics[width=0.9\textwidth]{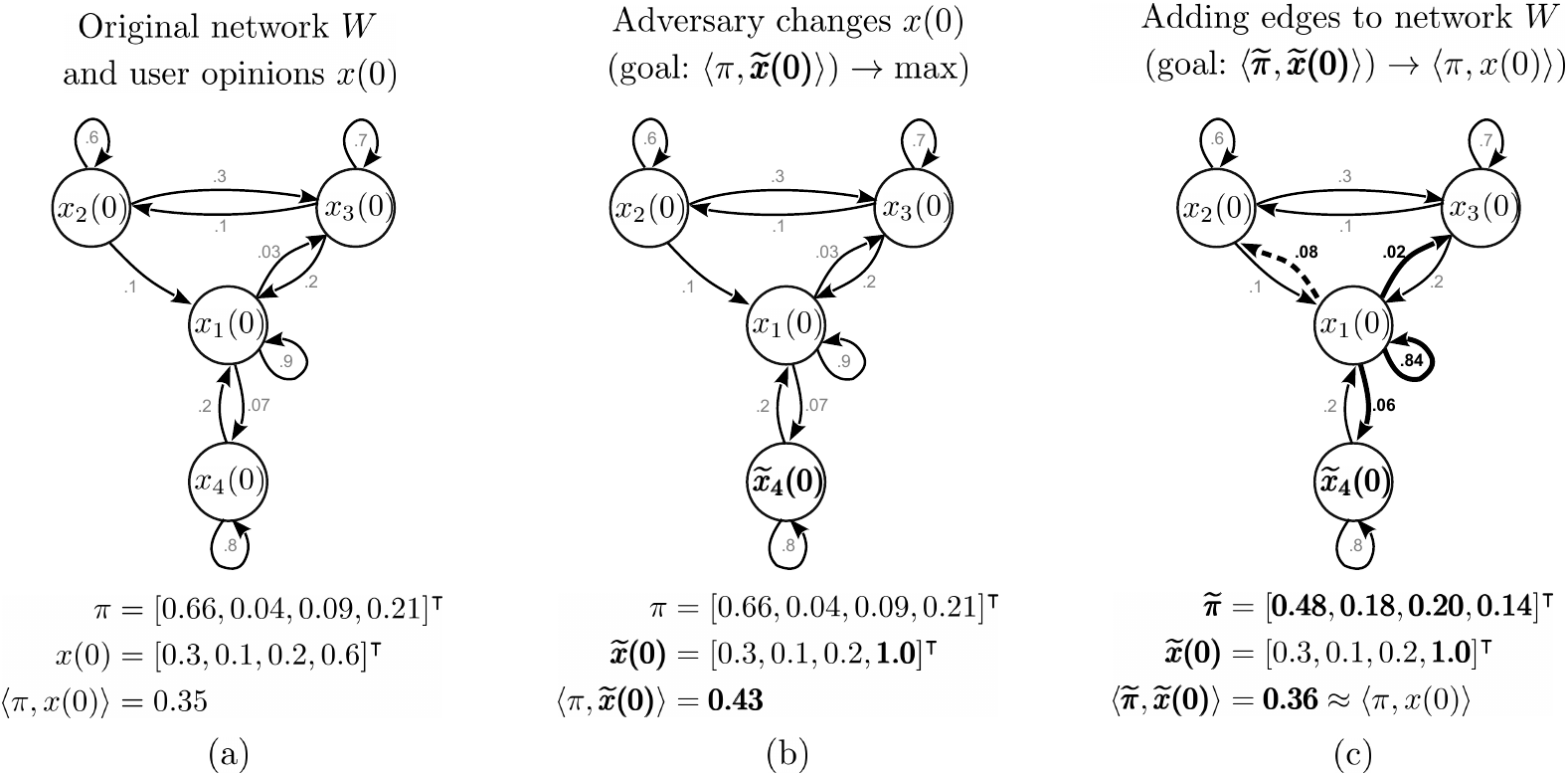}
    \caption{
        The adversary influences the users' opinions, $x(0) \to \xnew(0)$, increasing
        the asymptotic consensus value $\acv \to \acvnewx$.
        The network responds by adding edges, reducing the asymptotic consensus value
        $\acvnewxpi$, making it closer to the original $\acv$.
    }
    \label{fig:diver-by-example}
\end{figure*}
We assume that there is an adversary---an external party whose goal is to maximize the
asymptotic consensus value $\acv$. To that end, the adversary influences a
limited number of users in the network, changing their initial opinions and, thereby,
changing both the initial opinion distribution $x(0) \to \xnew(0)$ as well as
the asymptotic consensus value $\acv \to \acvnewx$, as shown in Fig.~\ref{fig:diver-by-example}b.
Our goal is to respond to this attack, and restore the asymptotic consensus value to
its original state \acv. However, we cannot directly influence the social network's
users' opinions; the only legitimate opinion control tool available to us is
\emph{edge recommendation}. We add a limited number of edges, thereby, changing the
distribution of eigenvector centralities $\pi \to \pinew$ and restoring the original
asymptotic consensus value, $\acvnewx \to \acvnewxpi \approx \acv$, as shown in
Fig.~\ref{fig:diver-by-example}c.

The \emph{central goal} of this work is to design a scalable algorithm that---under the
above described attack upon the opinion distribution---would identify the edges
whose addition to a social network would efficiently drive the asymptotic consensus value
to its state prior to the attack, disabling the latter's impact. Our \emph{specific contributions}
are:

\noindent \hspace{0.02in} \bull
We have defined \diver---a new problem of disabling external influence in a social
network via edge recommendation---and proven its NP-hardness.

\noindent \hspace{0.02in} \bull
We have provided novel perturbation analysis, having established how the network
        nodes' eigenvector centralities changes when a single edge is added to the network. This
        analysis has led to the definition of an edge score $f_\pi(i, j)$ that quantifies
        the potential impact of addition of directed edge $(i, j)$ to the network upon
        \diver's objective.

\noindent \hspace{0.02in} \bull
We have shown how to estimate edge scores $f_\pi$ in pseudo-constant time in networks
with skewed eigenvector centrality distribution, such as scale-free networks.
    
\noindent \hspace{0.02in} \bull
We have provided a pseudo-linear-time heuristic for \diver relying on edge
        scores $f_\pi$, and experimentally confirmed its effectiveness.

The paper is organized as follows.
Preliminaries and notation are provided in Sec.~\ref{sec:preliminaries}.
In Sec.~\ref{sec:problems-statement-and-hardness}, we formally define our
problem---\diver---and prove its NP-hardness.
A brief survey of literature on extremal network design and centrality perturbation
theory is given in Sec.~\ref{sec:background-work}.
Our main analyses and algorithms are developed in Sec.~\ref{sec:edge-addition}.
In particular, Sec.~\ref{sec:general-approach} provides an informal overview of our
approach towards solving \diver.
In Sec.~\ref{sec:edge-source-selection}, we address the problem of choosing a small
number of candidate edges out of the quadratic total number of candidates.
Sec.~\ref{sec:eigencentrality-vs-single-edge-perturbation} provides the analysis of
the impact of a single edge's addition upon the eigenvector centrality distribution.
Sec.\ref{sec:fpi} and Sec.~\ref{sec:fpi-efficient-comp} address the definition of
edge scores and their efficient computation, respectively.
The edge selection heuristic for \diver along with its time complexity are stated
in Sec.~\ref{sec:solving-diver}.
We conclude with experimental results in Sec.~\ref{sec:experiments}, and discussion
in Sec.~\ref{sec:discussion-and-future}.

\section{Preliminaries}
\label{sec:preliminaries}

We are given a sparse directed strongly connected aperiodic social network
$G(V, E)$, $|V| = n$, $|E| = \bigoh(n)$, having row-stochastic adjacency matrix
$W \in [0, 1]^{n \times n}$, $W \onebb = \onebb$---also known as the interpersonal
appraisal matrix---whose entry $w_{ij} \in [0, 1]$ reflects the relative extent to
which user $i$ takes into account the opinion of user $j$ while forming his or her
opinion. Aperiodicity can be replaced by the requirement of the network's having at
least one self-loop with a non-zero weight, which translates into a natural requirement
of having at least one user who does not completely disregard his or her own opinion
in the process of new opinion formation.

User opinions $x(t) \in [0, 1]^n$ at time
$t = 0, 1, 2, \dots$ are continuous, indicating an attitude towards a particular issue.
Given the initial user opinions $x(0) \in [0, 1]^n$, the opinions $x(t)$
evolve in discrete time as $x(t + 1) = W x(t)$, with each user's locally averaging
the opinions of all the users in his or her out-neighborhood, including the user's
own opinion.

In the remainder of the paper, we will mostly work with the initial
opinions $x(0)$, so we will use notation $x = x(0)$.

\begin{table}[t]
    \small
    \centering
    \begin{tabular}{|c|l|}
        \hline
        $\onebb$ & vector of all ones\\ \hline
        $\diag(v)$ & diagonal matrix with $v$ on the main diagonal\\ \hline
        $I$ & identity matrix\\ \hline
        $e_i$ & $i$'th column of the identity matrix\\ \hline
        $x(t)$ & (unaltered) user opinions at time $t$\\ \hline
        $x$ & $x(0)$\\ \hline
        $\xnew$ & altered user opinions at time $t = 0$\\\hline
        $n$ & number of nodes in the network\\ \hline
        $W$ & network's row-stochastic adjacency matrix\\ \hline
        $\Wnew$ & altered network' row-stochastic adjacency matrix\\ \hline
        $\edgeweight{ij}$ & weight of added directed edge $(i, j)$\\ \hline
        $\pi$ ($\pinew$) & $\ell_1$-normalized left dominant eigenvector of $W$ (of $\Wnew$)\\ \hline
        $m_{ij}$ & mean first passage time from $i$ to $j$ in chain $W$\\ \hline
    \end{tabular}
    \caption{Notation summary}
\label{tbl:notation}
\end{table}

Due to strong connectivity and aperiodicity of $G$, the opinion formation
process asymptotically converges, with
$\lim_{t \to \infty}{x(t)} = \acvs \onebb$,
where $\pi^\trans W = \pi^\trans$, $\|\pi\|_1 = 1$, so $\pi$ is the $\ell_1$-normalized
dominant left eigenvector of $W$. We refer to \acvs as the \emph{asymptotic consensus
value}---the opinion all the users are attracted to and asymptotically agree upon
under DeGroot model. By definition, $\pi$ is also a vector of eigenvector centralities
of the network's nodes, and can also be viewed as the users' no-teleportation PageRank
scores or the stationary distribution of the ergodic Markov chain with state transition
matrix $W^\trans$. Due to the latter, we may refer to $W$ as a Markov chain, and are
interested in the following properties of $W$ if viewed as such.

\begin{definition}[First passage time]
    The first passage time $T_{ij}$ from state $i$ to state $j$ of Markov chain $W$
    is a random variable describing the number of steps it takes for the chain started
    at state $i$ to reach state $j$. $T_{ii}$ is the first return time.
\end{definition}

\begin{definition}[Mean first passage time]
    The mean first passage time (MFPT) $m_{ij}$ from state $i$ to state $j$ of Markov
    chain $W$ is the expected first passage time $\expect{T_{ij} | \text{started at $i$}}$
    through state $j$ when the chain started at state $i$. $m_{ii}$ is the mean first return
    time (MFRT).
\end{definition}

The following two theorems immediately follow from Theorems 4.4.4 and 4.4.5 of Kemeny and Snell~\cite{kemeny1976finite}, respectively; regularity of the Markov chain $W$ in the
original theorems translates into our requirements of aperiodicity and strong connectivity
of the network with adjacency matrix $W$.

\begin{theorem}[Connection between MFRT and $\pi$]
    For any state $i$ of Markov chain $W$ with an aperiodic strongly
    connected network, $m_{ii} = 1 / \pi_i$.
    \label{thm:mfrt-vs-pi}
\end{theorem}

\begin{theorem}[MFPT one-hop conditioning]
    For any states $i$ and $j$ of Markov chain $W$ with an aperiodic strongly
    connected network, $m_{ij} = 1 + \sum_{k \neq j}{w_{ik}m_{kj}}$.
    \label{thm:mfpt-one-hop-cond}
\end{theorem}

\section{Problem's Statement and Hardness}
\label{sec:problems-statement-and-hardness}

Given a social network, at each point it time, we can observe its users' opinions.
We assume that, at some time point, an external adversary makes
an influence maximization attempt by targeting several users and changing their opinions,
with the goal of, w.l.o.g., maximizing the asymptotic consensus value.
Such influence attempts can be detected using opinion dynamics-aware anomaly detection
techniques~\cite{amelkin2017distance}. Alternatively, we can track whether the
current changes in user opinions follow a pattern prescribed by the solution of an
influence maximization problem\footnote{For DeGroot opinion dynamics model,
where the expression $\acvs \onebb$ for the asymptotic opinion distribution is linear
in $x$, influence maximization can be performed by solving an instance of a 0-1 knapsack.
The latter can be efficiently performed via dynamic programming in pseudo-linear time.}.

Having detected an external influence attempt, we are given the opinion
distribution $x \in [0, 1]^n$ preceding the attack as well as the externally altered
opinion distribution $\xnew \in [0, 1]^n$. As a result of the attack, the original
asymptotic consensus value \acvs has changed to \acvnewxs, where $\pi$, as before, is
the network's eigencentrality vector. Our goal is to add a limited number of edges to the
network and, thereby, change $\pi$ in such a way, that the resulting asymptotic consensus
value \acvnewxpis is close to its state \acvs before the attack. Formally, the problem of
disabling external influence via edge recommendation is defined as follows:
\begin{align}
    \boxed{
    \diver(W, k, x, \xnew) = {\argmin}_{\Wnew}{
        |
            \langle \pinew(\Wnew), \xnew \rangle -
            \langle \pi, x \rangle
        |
    },
    }
    \label{eq:opt-general}
\end{align}
where the perturbed row-stochastic adjacency matrix $\Wnew$ differs from $W$
by $k$ new edges whose weights $\edgeweight{ij}$ we cannot control (since these
weights correspond to the users' interpersonal appraisals), yet, can estimate
and, hence, assume the knowledge of. For $k = 1$, after addition of directed edge
$(r, c)$ with weight $\theta_{rc}$, $\Wnew$ will look as follows:
\begin{align}
    \Wnew = W - \edgeweight{rc} \diag(e_r) W + \edgeweight{rc} e_r e_c^\trans.
    \label{eq:single-edge-perturbation}
\end{align}
\vsc\vsbb
\begin{figure}[h!]
    \centering
    \includegraphics[width=1.25in]{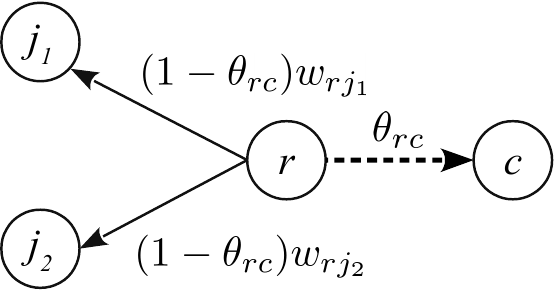}
    \label{fig:single-edge-perturbation}
\end{figure}
\vsa

In this work, we focus on deterministically adding edges with predefined weights
to the network, but our framework can be easily extended to the non-deterministic,
case with edge acceptance probabilities.
    	    
Complexity of \diver comes along two dimensions---the necessity to search for the
best subset of edges delivering the minimum of \diver's objective, and assessing
the impact of a given subset of edges upon the objective. While the latter can be
done in polynomial time\footnote{It narrows down to recomputing the dominant left
eigenvector of sparse $\Wnew$ perturbed with new edges, which can be done in
pseudo-linear time using the power method, where ``pseudo-'' reflects the dependency
of the power method's convergence rate upon the matrix' spectral gap.}, the edge
subset search cannot and is the cause of NP-hardness. In the following
Theorem~\ref{thm:hardness}, we formally show that, even for the case of an
undirected network, $\diver(W, k, x, \xnew)$ is NP-hard.

\begin{theorem}
    The problem $\diver(W, k, x, \xnew)$ of disabling external influence
    via $k$ edges' recommendation is NP-hard for undirected networks.
\label{thm:hardness}
\end{theorem}
\begin{proof}
    In the proof, we will show that $\diver(W, k, x, \xnew)$ applied to
    a certain simple undirected network can be used as a solver for the classic
    NP-complete subset sum problem. For readability, we will abuse notation
    and assume that the value of $\diver(W, k, x, \xnew)$ is the minimum itself,
    rather than the corresponding $\argmin$.
    
    \emph{1) Subset sum problems:}
    The subset sum problem $\ssp(\{z_i\}, s)$
    is a classic NP-complete problem of deciding whether a given finite set
    $\{z_i\} \subset \mathbb{Z}^n$ of integers has a non-empty subset with a predefined
    sum $s \in \mathbb{Z}$. (\ssp appears on Karp's list of NP-complete
    problems~\cite[p.95]{karp1972reducibility} under the name \knapsack).
    A related problem is the problem $\kssp(\{z_i\}, k, s)$ of deciding whether,
    among a finite number of bounded reals $z_i \in [0, 1]$, there is a non-empty
    subset of $k$ elements summing up to a given value $s \in [0, 1]$. Reduction
    $\ssp \propto \kssp$ is as follows:
    \begin{align*}
        \ssp(&\{z_i\}, s) = {\bigvee}_{k=1}^{n}{\kssp(\{z''_i\}, k, s'')},\\[0.04in]
        z'_i &= z_i + L \in \mathbb{Z}_+,\ s' = s + kL \in \mathbb{Z}_+,\\
        L &= |\min\{0, \min\{s, \min{z_i}\}\}|\\
        z''_i &= z'_i / M \in [0, 1],\ s''_i = (s + kL) / M \in [0, 1],\\
        M &= \max\{s', \max{z'_i}\}.
    \end{align*}
    
    \emph{2) Undirected uniformly weighted networks and their eigenvector centrality:}
        Let $W_{01} \in \{0, 1\}^{n \times n}$ be the binary adjacency matrix of an undirected network, $d = W_{01} \onebb$ be a vector of node degrees,
        and $D = \diag(d)$. Further, let $W = D^{-1} W_{01}$. We say that $W$ is the
        adjacency matrix of an \emph{undirected uniformly weighted network} (since,
        all the edges within the same neighborhood are weighted equally). Notice that
        $W$ is row-stochastic, as $W \onebb = D^{-1} W_{01} \onebb = D^{-1} d = \onebb$.
        
        Since $d^\trans W = d^\trans D^{-1} W_{01} = \onebb^\trans W_{01} = d^\trans$,
        vector $\pi = d / \|d\|_1 = d / (2|E|) = d / (2m)$ is the $\ell_1$-normalized dominant left
        eigenvector--or, eigenvector centrality---of $W$. If the underlying unweighted
        network $W_{01}$ is perturbed with $k$ undirected edges $(i, j) \in S$, $|S| = k$,
        then the eigenvector centrality of the corresponding weighted network becomes
        \begin{align}
            \pinew
                &= \frac{1}{2(m + k)} \left(
                    d + {\sum}_{(i, j) \in S}{(e_i + e_j)}
                \right) \notag \\
                &= \frac{1}{m + k}\left(
                    m\pi + {\sum}_{(i, j) \in S}{(e_i + e_j) / 2}
                \right),
            \label{eq:uw-pinew}
        \end{align}
        where $e_i$ is the $i$'th column of the identity matrix.

    \emph{3) \diver in undirected uniformly weighted networks:}
    If network $W$ is undirected uniformly weighted and, thus,
    defined by its binary adjacency matrix $W_{01}$, then \diver's objective function
    over such $W$ can be rewritten as follows:
    \begin{align*}
        f(\Wbinnew)
            &= |\langle \pinew, \xnew \rangle - \langle \pi, x \rangle| =
            (\text{from~(\ref{eq:uw-pinew})}) =\\
            &= \left|
                \frac{m}{m + k} \langle \pi, \xnew \rangle +
                \frac{1}{2(m+k)} \langle \sum\limits_{(i, j) \in S}{(e_i + e_j)}, \xnew \rangle -
                \langle \pi, x \rangle
                \right|\\
            &= \frac{1}{m + k}
                \left|
                \langle \sum\limits_{(i, j) \in S}{(e_i + e_j)}, \xnew / 2 \rangle
                -
                \langle \pi, (m+k)x - m \xnew \rangle
                \right|\\
            &= a(k)
                \left|
                    \langle \sum\limits_{(i, j) \in S}{(e_i + e_j)}, \xnew / 2 \rangle
                    -
                    b(k, x, \xnew)
                \right|,
    \end{align*}
    where $b(k, x, \xnew) = \langle \pi, (m+k)x - m\xnew \rangle$.
    Since $k$, and, consequently, $a(k)$ are constant, minimization of $f(\Wbinnew)$
    is equivalent to minimization of
    \begin{align}
        f'(\Wbinnew) =
            \left|
                \langle \sum\limits_{(i, j) \in S}{(e_i + e_j)}, \xnew / 2 \rangle -
                b(k, x, \xnew)
            \right|.
        \label{eq:uw-obj}
    \end{align}

    \emph{4) Reduction $\kssp \propto \diver$:}
    Suppose we are given an instance $\kssp(z, k, s)$, with $z \in [0, 1]^n$,
    $k \in \mathbb{N}$, and $s \in [0, 1]$. In what follows, we will show that
    the solution to $\kssp(z, k, s)$ is obtained by checking whether
    \begin{align}
        \min_{\Wkcnew}{
            \diver\left(
                \Wkc,
                k,
                \frac{s\onebb + m (z \otimes \onebb_2)}{m + k},
                z \otimes \onebb_2
            \right)
        } = 0,
        \label{eq:kssp-diver-reduction}
    \end{align}
    where $\Wkc$ is the adjacency matrix of an undirected uniformly weighted $2n$-clique from
    which edges $C = \{ (2i - 1, 2i) \mid i = 1, \dots, n\}$ have been removed (see Fig.~\ref{fig:wkc}), $\Wkcnew$ is $\Wkc$ perturbed with $k$ edges $S = \{(2i - 1, 2i)\} \subseteq C$,
    $\onebb = \onebb_{2n}$, and $\otimes$ is Kronecker product.
    \begin{figure}[th]
        \centering
        \includegraphics[width=1.5in]{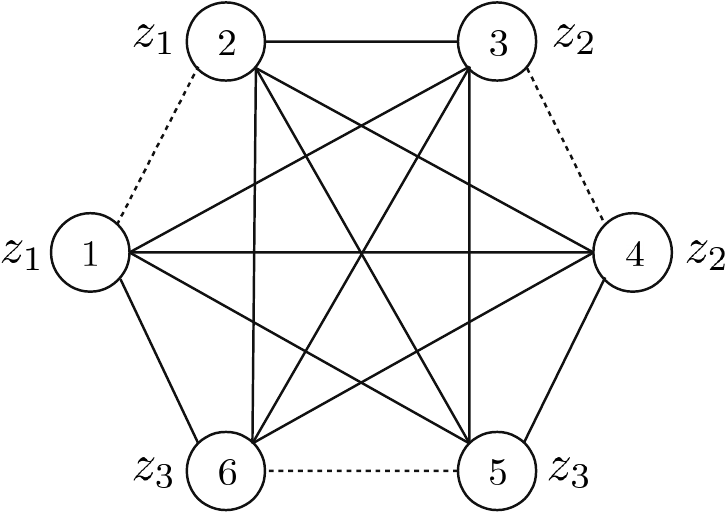}
        \caption{
            Network $\Wkc$ for $n = 3$; absent edges $C$ are displayed dashed. Node states
            $z \otimes \onebb_2$, used in the reduction, are displayed next to the nodes.
        }
        \label{fig:wkc}
    \end{figure}
    
    It is easy to show that the proposed input to \diver is indeed legal ($\Wkc$
    is row-stochastic matrix of a uniformly weighted undirected strongly connected
    aperiodic network; and $\xnew = z \otimes \onebb_2$ and 
    $x = \frac{s\onebb + m (z \otimes \onebb_2)}{m + k}$ are legal 
    vectors of altered and original user opinions, respectively.
    
    Let us show what \diver transforms into under the proposed input of~(\ref{fig:wkc}). First,
    we notice that, for $b(k, x, \xnew) = \langle \pi, (m+k)x - m\xnew \rangle$ of~(\ref{eq:uw-obj}),
    the following holds
    \begin{align*}
        b&\left(
            k,
            \frac{s\onebb + m \xnew}{m + k},
            \xnew
        \right) =
        \langle
        \pi, s\onebb + m \xnew - m \xnew
        \rangle = s \langle \pi, \onebb \rangle = s.
    \end{align*}
    Then, \diver's objective~(\ref{eq:uw-obj}) under input~(\ref{eq:kssp-diver-reduction}) will look as
    \begin{align*}
        f'(\Wkcnew) = 
            \left|
                \left\langle \sum\limits_{(i, j) \in S}{(e_i + e_j)}, \frac{z \otimes \onebb_2}{2} \right\rangle -
                s
            \right|
            = \left|
                \sum\limits_{\ell = 1}^{n}{y_\ell z_\ell} - s
            \right|,
    \end{align*}
    where $y_{\ell}$ are edge decision variables
    $$
    y_{\ell} = \begin{cases}
        1, & \text{if } (2\ell - 1, 2\ell) \in S,\\
        0, & \text{otherwise}.
    \end{cases}
    $$
    Thus, solving \diver via minimizing $f'(\Wkcnew)$, we look for a subset of
    $\{z_\ell\}$ of size $k$ summing up to $s$, which is exactly what \kssp is after,
    so $\kssp \propto \diver$.

    Parts 1) and 4) of the proof together establish $\ssp \propto \kssp \propto \diver$,
    so \diver is NP-hard.
    \qedbull
\end{proof}

\section{Background Work}
\label{sec:background-work}

\diver is, essentially, a problem of strategically modifying eigenvector centrality
$\pi \to \pinew$ of a directed weighted network via edge addition, with the goal of
optimizing the absolute value of a linear function of $\pinew$. While
this problem is new, there is a range of related problems---in extremal network design
as well as in the perturbation analysis of centrality measures and stationary distributions
of Markov chains---related to ours either in the nature of the objective being optimized
or the methods and analyses used. We survey several groups of these works in the following
subsections.

\subsection{Analytic Optimization of Network Topology}
\label{sec:analytic-opt-of-network-topology}

The first class of related works are the network design problems, where a network's
topology is altered to optimize some property of that network. Both the optimized property
and the methods involved in the solution are analytic (in contrast to combinatorial
goals and methods, reviewed separately).

\emph{Algebraic Connectivity:}
Ghosh and Boyd~\cite{ghosh2006growing} studied the problem of maximizing the
algebraic connectivity---the second smallest eigenvalue of the combinatorial
Laplacian $L$~\cite{merris1994laplacian}---of an undirected unweighted network
via edge addition. If $a_i$ is the $i$'th column
of the network's incidence matrix, then the optimization problem being addressed is
$$
    \lambda_2(L + \sum{x_i a_i a_i^\trans}) \to \max,\ x^\trans \onebb = const,\ x \in \{0, 1\}^n.
$$
The authors formulate the problem as a semidefinite program (SDP) via convex relaxation
($x \in [0, 1]^n$), which is feasible to solve for small networks. They also
provide a greedy perturbation heuristic that picks edges $(i, j)$ based on the largest
value of $(v_i - v_j)^2$---the squared difference of the Fiedler vector's components
corresponding to each edge's ends. The authors show that, in case of simple $\lambda_2$,
value $(v_i - v_j)^2$ gives the first-order approximation of the increase in $\lambda_2(L)$
if edge $(i, j)$ is added to the network. The heuristic outperforms the SDP solution in
experiments on synthetic data. The authors also derive bounds on algebraic connectivity
under single-edge perturbation. More recently, this approach has
been employed by Yu et al.~\cite{yu2015friend} for the design of an edge selection
heuristic that the authors have augmented with an extra objective---neighborhood
overlap-based user similarity (which likely correlates with edge acceptance
likelihood).

\emph{Spectral Radius:}
Van  Mieghem et al.~\cite{vanmieghem2011decreasing} study the problem of minimizing the
spectral radius of an undirected network via edge or node removal. They prove
NP-hardness of the problem, and show that the edge selection heuristic that picks edges
$(i, j)$ with the largest scores $v_i v_j$---where $v$ is the dominant
eigenvector---performs well in experiments. More recently, Saha et
al.~\cite{saha2015approximation} addressed the same problem of spectral radius minimization,
and designed a walk-based algorithm, relying on the link between the sum of powers of
eigenvalues of a network and the number of closed walks in it, and provided approximation
guarantees for them. Zhang et al.~\cite{zhang2015controlling} studied spectral radius
minimization for directed networks under SIR model, and provided an SDP/LP-based
solutions, having high polynomial time complexity.

\emph{Eigenvalues and Their Functions:}
Tong et al.~\cite{tong2012gelling} investigate how to optimize the diffusion
rate---expressed via the largest eigenvalue of the adjacency matrix---through
a directed strongly connected unweighted (see~\cite{chen2016eigen} for the weighted
case's treatment) network via edge addition or removal. Similarly to
Van  Mieghem et al.~\cite{vanmieghem2011decreasing}, the authors use first-order
perturbation theory in order to assess the effect of the deletion of $k$ edges
\begin{align*}
    \lambda_{max} - \widetilde{\lambda}_{max}
        = \sum{u_i v_j} / \langle u, v \rangle + \bigoh(k),
\end{align*}
where $u$ and $v$ and the left and right dominant eigenvectors of the adjacency matrix,
respectively. This analysis inspires an edge selection heuristic, with the quality
of edge $(i, j)$'s being defined as $u_i v_j$, similarly to $v_i v_j$ edge
score of~\cite{vanmieghem2011decreasing}. Le et al.~\cite{le2015met} extend this
result to the networks with small eigengaps. The key idea of their approach is tracking
multiple (instead of just the dominant) eigenvalues of the network.
Chan et al.~\cite{chan2014make} target optimization of natural connectivity---a network
robustness measure defined, roughly, as an average of exponentiated eigenvalues of the
adjacency matrix---of an undirected strongly connected network via the change of its
topology. For edge addition, they focus on a small number of candidate edges whose
both ends have high eigencentrality.

\emph{Other Objectives:}
The SDP-based approach of Ghosh and Boyd~\cite{ghosh2008minimizing} has been applied by
the same authors to minimization of the total effective resistance of an undirected
electric network via edge weight selection. Arrigo and Benzi~\cite{arrigo2016updating}
address the problem of optimizing the total communicability---the sum of the entries
in the exponential of the adjacency matrix---in an undirected connected network via
edge addition and removal. The authors use edge selection heuristics, favoring edges
between the nodes having high eigenvector centrality (for edge addition) or edges between
the nodes having a large sum of their degrees (for edge removal). Garimella et al.~\cite{garimella2017reducing} study an edge recommendation problem targeting reduction
of polarization in a directed unweighted network, where polarization is measured via
a random walk-based score. The edges are created between users ``holding opposing views''.
Similarly to~\cite{chan2014make,arrigo2016updating}, the authors use an edge-selection
heuristic that favors edges between high-degree nodes.

\subsection{Combinatorial Optimization of Network Topology}

These works address network design problems whose objectives or methods are of
combinatorial nature. A large portion of these
works are dedicated to direct information spread optimization in combinatorial
opinion dynamics models, in contrast to indirectly optimizing some analytic
feature of the network, such as the spectral radius of its adjacency matrix,
expected to facilitate or hinder information propagation.

\emph{Information Spread:}
Chaoji et al.~\cite{chaoji2012recommendations} look at a problem of maximizing the
size of the activated user set under the Independent Cascade-like opinion dynamics
model in an undirected network via edge addition. The authors prove NP-hardness of
the problem, apply continuous relaxation to gain submodularity of the objective,
and design a greedy cubic-time approximation algorithm for the relaxed problem.
Kuhlman et al.~\cite{kuhlman2013blocking} focus on general threshold-based propagation
models, and address the problem of minimizing the contagion spread via edge deletion
in a directed weighted network. The authors prove inapproximability of the problem,
and design a spread simulation-based heuristic, that proves to be effective in
experiments. The work of Khalil et al.~\cite{khalil2014scalable} is dedicated to
facilitating or hindering the spread of information under Linear Threshold (LT)
model via edge addition or deletion in a directed weighted network. The authors design an influence objective function
and prove its supermodularity. The latter property used together with sampling
of LT process realizations allows for the design of an efficient linear-time
algorithm for target edge selection.

\emph{Shortest Paths and Optimal Flows:}
Phillips~\cite{phillips1993network} studied the problem of minimizing a combinatorial
maximum flow / minimum cut in a network. Each capacitated edge has
a destruction cost, and the adversary needs to select a subset of edges to destroy,
constrained by the total edge destruction budget. The authors prove NP-hardness
of the problem, and design an FPTAS for the case of a planar network. Israeli and
Wood~\cite{israeli2002shortest} conduct a study of an NP-hard problem of maximizing
a single $s$-$t$ shortest path via edge removal in a directed network, formulated as
a mixed-integer program (MIP). Due to the prohibitive time complexity of a direct
solution of a MIP problem, the authors propose several decomposition techniques
to accelerate the computation under some assumptions on the edge removal delays.
Papagelis et al.~\cite{papagelis2011suggesting} address the problem of minimizing
the average all-pairs shortest path length in a connected undirected network via
edge addition, and propose a greedy algorithm and two heuristics. Their most
efficient algorithm has a quadratic time complexity. Ishakian et
al.~\cite{ishakian2012framework} define a general path-counting centrality measure
and study a problem of maximizing the centrality of a given node via edge addition
in a DAG. The authors use a quadratic-time greedy strategy
for picking edges providing the largest marginal increase of the objective.
Parotsidis et al.~\cite{parotsidis2016centrality} study minimization of the sum
of lengths of the shortest paths from a target node to all other nodes via link
recommendation to the target node in an undirected network. The problem is proven
to be NP-hard, and an efficient approximation algorithm is designed, employing
submodularity of the objective. A related problem of minimizing the maximal
shortest path length has been previously addressed by Perumal et al.~\cite{perumal2013minimizing};
another related problem of maximizing the coverage centrality---the number of
unique node pairs whose shortest paths pass through a given node---is addressed
by Medya et al.~\cite{medya2017maximizing}.

\subsection{Centrality Perturbation and Manipulation}

These works study either how eigenvector centrality or PageRank or the stationary
distribution of a Markov chain change when a network's structure is perturbed;
or how to strategically manipulate centrality by altering the network.

\emph{Strategic Centrality Manipulation:}
Avrachenkov and Litvak~\cite{avrachenkov2006effect} analyze to what extent a node can
improve its PageRank by creating new outgoing edges. The authors derive equalities that result
in a conclusion that the PageRank of a web-page cannot be considerably improved by
manipulating its outgoing edges. The authors also derive an optimal linking strategy, stating
that it is optimal for a web-page to have only one outgoing edge pointing to a web-page with the
shortest mean first passage time back to the original page. We can come to similar conclusions
for eigenvector centrality in an arbitrarily weighted network using
Theorem~\ref{thm:single-edge-perturbation} from 
Sec.~\ref{sec:eigencentrality-vs-single-edge-perturbation} of our work.
De Kerchove et al.~\cite{de2008maximizing} generalize the results of Avrachenkov and
Litvak~\cite{avrachenkov2006effect}, studying maximization of the sum of PageRanks of
a subset of nodes via adding outgoing edges to them.
Cs{\'a}ji et al.~\cite{csaji2014pagerank} study the problem---originally, posed
by Ishii and Tempo~\cite{ishii2009computing}---of optimizing the PageRank of a given node
via edge addition in a directed network. The authors formulate the optimization problem
as a Markov decision process and propose a (generally, not scalable) polynomial-time
algorithm for it. More recently, Ye et al.~\cite{ye2017modification} studied the problem
of reducing the social dominance of the central node in a star network in the context
of a Friedkin-Johnsen model defined for issue sequences~\cite{jia2015opinion} via
structural modifications of the network and, in particular, via edge addition. By
exploiting regularity of a star network's structure, the authors establish the conditions
under which the social dominance can shift from the center to one of the peripheral nodes.

\emph{Centrality Perturbation Analysis:}
Cho and Meyer~\cite{cho2000markov} provide coarse bounds for the stationary distribution
of a generally perturbed Markov chain expressed via MFPTs:
$$
    |\pi_i - \pinew_i| / \pi_i \leq \|E\|_{\infty} \max\limits_{i \neq j}{m_{ij}} / 2,
$$
where $E$ is an additive perturbation of the state transition matrix.
Chien et al.~\cite{chien2001towards} provide an efficient algorithm for incremental
computation of PageRank over an evolving edge-perturbed graph, with the analysis'
drawing upon the theory of Markov chains. The key idea of their algorithm is to
contract the network and localize its part where the nodes are likely to have changed
their PageRank scores under the perturbation.
Jeh and Widom~\cite{jeh2003scaling} study incremental computation of personalized PageRank.
Langville and Meyer~\cite{langville2006updating} provide exact equalities for the change
in the stationary distribution of a perturbed Markov chain using group inverses. They
address the problem of updating the
stationary distribution under multi-row perturbation via exact and approximate
aggregation, similarly to what Chien et al.~\cite{chien2001towards} did for PageRank.
Hunter~\cite{hunter2005stationary} addresses the same problem of establishing equalities
for the change in the stationary distribution, yet, provides an answer that does not
involve group inverses and, instead, uses mean first passage times in a Markov chain;
our perturbation analysis in Sec.~\ref{sec:eigencentrality-vs-single-edge-perturbation}
builds upon this result.
Como and Fagnani~\cite{como2015robustness} provide an upper bound on the perturbation
of the stationary distribution of a Markov chain in terms of the mixing time of the
chain as well as the entrance and escape likelihoods to and from the states with
perturbed out-neighborhoods.
Bahmani et al.~\cite{bahmani2012pagerank} address the problem of updating PageRank
algorithmically. The proposed node probing-based algorithms provide a close estimate
of the network's PageRank vector by crawling a small portion of the network.
More recently, Li et al.~\cite{li2015cheetah} and Chen and Tong~\cite{chen2017eigen}
addressed a general problem of updating eigenpairs of an evolving network. Chen and
Tong provide a linear-time algorithm for tracking top eigenpairs. Finally, there
are works on updating non-spectral centrality measures, such as betweenness~\cite{lee2012qube}
and closeness~\cite{sariyuce2013incremental}.

\section{Strategic Edge Addition to the Network}
\label{sec:edge-addition}

\subsection{Overview of the General Approach}
\label{sec:general-approach}

Since, according to Theorem~\ref{thm:hardness}, \diver optimization problem
\begin{align}
    \diver(W, k, x, \xnew) = {\argmin}_{\Wnew}{
        |
            \langle \pinew(\Wnew), \xnew \rangle -
            \langle \pi, x \rangle
        |
    },
    \tag{\ref{eq:opt-general}}
\end{align}
is NP-hard, we need to design a heuristic for it. 
Our general approach---formalized later in Sec.~\ref{sec:solving-diver}---is as follows.
We will assess candidate edges with respect to how much their addition to the network
can decrease term \acvnewxpis of (\ref{eq:opt-general}), and,
then, iteratively add the most promising edges to the network until we are satisfied
with the value of \diver's objective.

Thus, our foremost concerns now are the selection of a small number of candidate edges to assess
and the subsequent assessment of the potential impact of these edges' addition to the network upon the
network's eigenvector centrality. They are addressed in the following two sections.

\subsection{Selection of Candidate Edge Source Nodes}
\label{sec:edge-source-selection}

The general approach of Sec.~\ref{sec:general-approach} involves assessing candidate
edges individually. However, the number of absent edges in a sparse network is $\bigoh(n^2)$,
and inspecting all of them is unfeasible for large networks. Hence, we will focus on
a small number of candidate edges, outgoing from $n_{src} = const \ll n$ network nodes.
The latter implies that a small number of nodes are being the sources for most---or, at
least, a large number of---``good'' candidate edges in the network. Intuitively, the
nodes having the largest (eigenvector) centrality should be those edge sources; the
changes in their out-neighborhoods should have the largest impact upon the centrality
distribution in the network, as Fig.~\ref{fig:src-centrality-vs-acv-reduction} suggests.
This intuition will find formal support in Corollary~\ref{thm:small-nsrc} of Theorem~\ref{thm:single-edge-perturbation} in the following
Sec.~\ref{sec:eigencentrality-vs-single-edge-perturbation}.
\begin{figure}[ht]
    \centering
    \includegraphics[width=0.85\linewidth]{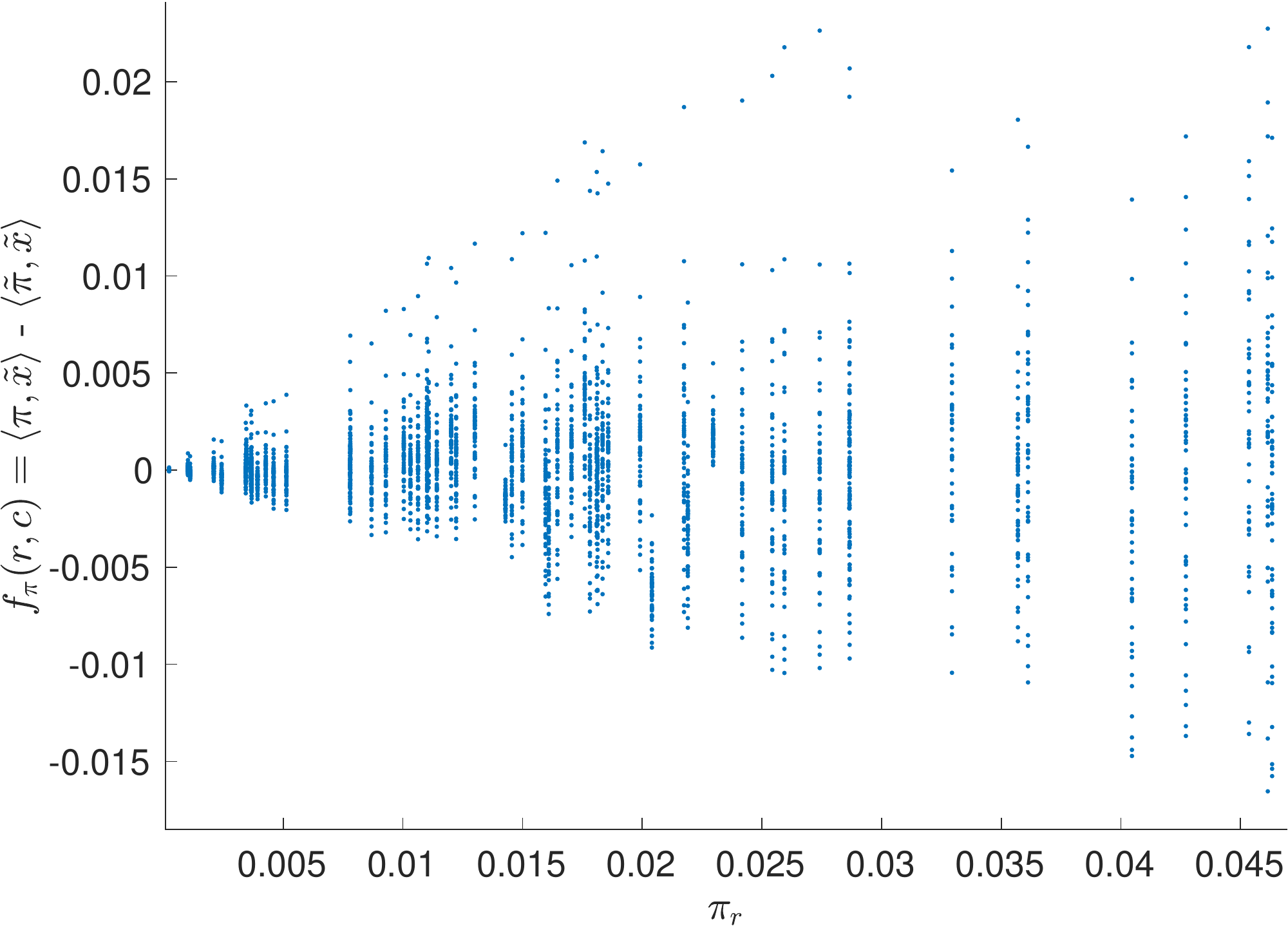}
    \caption{
        Dependence of the asymptotic consensus value reduction
        $f_\pi(r, c) = \langle \pi, \xnew \rangle - \langle \pinew, \xnew \rangle$ after
        addition of edge $(r, c)$, $\edgeweight{rc} = const$ to a scale-free network
        $(n = 100, \gamma = -2.5)$ on the edge source's eigenvector centrality $\pi_r$.
    }
    \label{fig:src-centrality-vs-acv-reduction}
\end{figure}

Consequently, to make sure that a candidate edge's addition to the network has a large
impact---either positive or negative---upon the asymptotic consensus value, we can select
candidate edges outgoing from high-centrality nodes. Fortunately, the number of
such nodes in real-world social networks is indeed small, and most nodes are at the
periphery, which justifies our choice of $n_{src} = const \ll n$.

\subsection{Eigencentrality Under Single-edge Perturbation}
\label{sec:eigencentrality-vs-single-edge-perturbation}

In order to tackle \diver~(\ref{eq:opt-general}), we need to understand how the
addition of a single edge $\widetilde{w}_{rc} = \edgeweight{rc} \in (0, 1]$ from
node $r$ to node $c$ in the network affects the eigencentrality vector $\pinew$. 
We assume that edge $(r, c)$ is originally absent, $w_{rc} = 0$,
and use the same single-edge perturbation model
\begin{align}
    \Wnew = W - \edgeweight{rc} \diag(e_r) W + \edgeweight{rc} e_r e_c^\trans \tag{\ref{eq:single-edge-perturbation}}
\end{align}

In our subsequent perturbation analysis, we will make the following Assumption~\ref{thm:rational-selfishness-assumption}.
\begin{assumption}[Rational Selfishness]
    Let us assume that the users are rationally selfish in that for any
    user $i$, $\forall j \neq i: w_{ii} > w_{ij}$. Thus, each user trusts
    his or her own opinion more than the opinion of any other user.
    \label{thm:rational-selfishness-assumption}
\end{assumption}

The following Theorem~\ref{thm:single-edge-perturbation} states how the eigencentrality
vector changes under a single-edge perturbation~(\ref{eq:single-edge-perturbation}).
\begin{theorem}[Single-Edge Perturbation]
    Under Assumption~\ref{thm:rational-selfishness-assumption}, for a single-edge
    perturbation~(\ref{eq:single-edge-perturbation}) of a strongly connected aperiodic
    network with adjacency matrix $W$, the network's eigenvector centrality
    changes as
    \begin{align}
        \pinew_j = \pi_j \left[
            1 - \frac{
                \edgeweight{rc} (m_{cj} (1 - \delta\{j, c\}) - m_{rj} + 1)
            }{
                m_{rr} + \edgeweight{rc} (m_{cr} - m_{rr} + 1)
            }
        \right],
    \label{eq:single-edge-perturbation-pinew}
    \end{align}
    where $m_{ij}$ is the mean first passage time from state $i$ to $j$
    of Markov chain $W$, and $\delta$ is Kronecker delta. In particular,
    \begin{align}
        \pinew_r &= 1 / [m_{rr} + \edgeweight{rc} (m_{cr} - m_{rr} + 1)].
            \label{eq:single-edge-perturbation-pinew-r} 
    \end{align}
    \label{thm:single-edge-perturbation}
\end{theorem}
The proof of Theorem~\ref{thm:single-edge-perturbation} will rely on the perturbation
result of Hunter~\cite{hunter2005stationary}, provided for reference as
Theorem~\ref{thm:hunter-single-row-perturbation} below.
\begin{theorem}[\citex{hunter2005stationary}{Theorem 4.4}]
    Suppose multiple perturbations occur in $r$'th row of $W$. Let $\epsilon_i = \Wnew_{ri} - W_{ri}$,
    the minimal negative perturbation happen at state $a$, with $\epsilon_a = -m = \min{\{\epsilon_j \mid
    1 \leq j \leq n\}}$, and the maximal positive perturbation occur at state $b$ with
    $\epsilon_b = M = \max{\{\epsilon_j \mid 1 \leq j \leq n\}}$. Also, let $P$ be the set of
    positive perturbation indices, excluding $b$, and $N$ be the set of negative perturbation
    indices, excluding $a$. Then,
    \begin{align*}
        \pi_j - \pinew_j =
            \begin{cases}
            \pi_a \pinew_r [
                M m_{ba} +
                \sum\limits_{k \in P \cup N}{\epsilon_k m_{ka}}
            ] & \text{if } j = a,\\
            \pi_b \pinew_r [
                -m m_{ab} +
                \sum\limits_{k \in P \cup N}{\epsilon_k m_{kb}}
            ] & \text{if } j = b,\\
            \pi_j \pinew_r [
                -m m_{aj} + M m_{bj} +
                \sum\limits_{\substack{k \in P \cup N \\k \neq j}}{\epsilon_k m_{kj}}
            ] & \text{if } j \neq a, b.
        \end{cases}
    \end{align*}
    \label{thm:hunter-single-row-perturbation}
\end{theorem}
\begin{proof}(Theorem~\ref{thm:single-edge-perturbation})
    Let us apply Theorem~\ref{thm:hunter-single-row-perturbation} to our case
    of a single-edge perturbation~(\ref{eq:single-edge-perturbation}).
    We are adding edge $(r, c)$ with weight $\edgeweight{rc}$ to the network.
    Due to the form~(\ref{eq:single-edge-perturbation}) of our single-edge perturbation,
    the only positive perturbation occurs at the added edge's destination node $c$, so $b = c$,
    $\epsilon_c = M = \edgeweight{rc}$, and $P = \emptyset$.
    For all the other out-neighbors $i$ of the new edge's source node $r$, the corresponding
    perturbations $\epsilon_i = -\edgeweight{rc} w_{ri}$ are negative.
    Due to Assumption~\ref{thm:rational-selfishness-assumption}, $\forall i \neq r: w_{rr} > w_{ri}$,
    so the minimal negative perturbation occurs at $i = r$, and, thus, $a = r$ and
    $\epsilon_a = -m = -\edgeweight{rc} w_{rr}$.
    
    Let us first show the validity of~(\ref{eq:single-edge-perturbation-pinew}) in case
    of $j = r$, that is, (\ref{eq:single-edge-perturbation-pinew-r}). According to
    Theorem~\ref{thm:hunter-single-row-perturbation},
    \begin{align*}
        \pi_r - \pinew_r
            &= \pi_r \pinew_r \left[
                    \edgeweight{rc} m_{cr} +
                    \sum\limits_{k \in P \cup N}{(-\edgeweight{rc}w_{rk})m_{kr}}
                \right]\\
            &= (\text{as $w_{rc} = 0$}) = \edgeweight{rc} \pi_r \pinew_r \left[
                    m_{cr} - \sum\limits_{k \neq r}{w_{rk} m_{kr}}
                \right].
    \end{align*}
    Using the one-hop conditioning Theorem~\ref{thm:mfpt-one-hop-cond}, the obtained
    expression can be written as
    \begin{align*}
        \pi_r - \pinew_r &= \edgeweight{rc} \pi_r \pinew_r \left[ m_{cr} - m_{rr} + 1 \right]\\[0.04in]
        \iff \pinew_r &= \pi_r / (1 + \edgeweight{rc} \pi_r [m_{cr} - m_{rr} + 1]).
    \end{align*}
    Dividing the numerator and denominator in the right-hand side of the obtained
    expression by $\pi_r > 0$ and using equality $1 / \pi_r = m_{rr}$ from
    Theorem~\ref{thm:mfrt-vs-pi}, we obtain~(\ref{eq:single-edge-perturbation-pinew-r}).
    
    Let us similarly show the validity of the case $j \neq r, c$. From
    Theorem~\ref{thm:hunter-single-row-perturbation},
    \begin{align*}
        \pi_j - \pinew_j
            &= \pi_j \pinew_r \left[
                -\edgeweight{rc} w_{rr} m_{rj}
                + \edgeweight{rc} m_{cj}
                + \sum\limits_{\substack{k \in P \cup N \\k \neq j}}{
                    (-\edgeweight{rc}w_{rk})m_{kj}
                }
            \right]\\
            &= \edgeweight{rc} \pi_j \pinew_r \left[
                m_{cj}
                - w_{rr} m_{rj}
                - \sum\limits_{k \neq r, c, j}{
                    w_{rk} m_{kj}
                }
            \right]\\
            &= (\text{as $w_{rc} = 0$}) = \edgeweight{rc} \pi_j \pinew_r \left[
                m_{cj}
                - \sum\limits_{k \neq j}{ w_{rk} m_{kj} }
            \right]\\
            &\iff \text{(from Theorem~\ref{thm:mfpt-one-hop-cond})}\\[0.1in]
            &\iff \pinew_j = \pi_j [1 - \edgeweight{rc} \pinew_r (m_{cj} - m_{rj} + 1)].
    \end{align*}
    Substituting~(\ref{eq:single-edge-perturbation-pinew-r}) in the obtained expression,
    we get~(\ref{eq:single-edge-perturbation-pinew}) for $j \neq c$. The proof for case
    $j = c$ is similar and, hence, is omitted.
    \qedbull
\end{proof}

The following Corollary~\ref{thm:small-nsrc}---justifying Sec.~\ref{fig:src-centrality-vs-acv-reduction}'s
focus on top-centrality edge source nodes---immediately follows from
equation~(\ref{eq:single-edge-perturbation-pinew}) of Theorem~\ref{thm:single-edge-perturbation}
used with Theorem~\ref{thm:mfrt-vs-pi}.
\begin{corollary}
    Under perturbation~(\ref{eq:single-edge-perturbation}) of the network
    with a single edge $(r, c)$, $\edgeweight{rc} > 0$, it holds that
    $\lim_{\pi_r \to 0}{\pinew} = \pi$, and, thus,
    $\lim_{\pi_r \to 0}{f_\pi(r, c)} = \lim_{\pi_r \to 0}{(\acvnewxs - \acvnewxpis)} = 0$.
    \label{thm:small-nsrc}
\end{corollary}

\subsection{Asymptotic Consensus Value Under Single-Edge Perturbation}
\label{sec:fpi}

To solve \diver, we are interested in adding candidate edges that would result in a large
reduction $f_\pi(r, c) = \langle \pi, \xnew \rangle - \langle \pinew, \xnew \rangle$ of the asymptotic
consensus value. While Theorem~\ref{thm:single-edge-perturbation} states how different
components of the eigencentrality vector change under a single-edge perturbation~(\ref{eq:single-edge-perturbation}), the
following Theorem~\ref{thm:edge-score} is concerned with the effect of such perturbation upon
the value of $f_\pi(r, c)$.
\begin{theorem}
    Under the rational selfishness Assumption~\ref{thm:rational-selfishness-assumption},
    for a single-edge perturbation~(\ref{eq:single-edge-perturbation}) of $W$, the
    asymptotic consensus value $\langle \pi, \xnew \rangle$ decreases as follows:
    \begin{align}
        f_{\pi}(r, c) &= \langle \pi, \xnew \rangle - \langle \pinew, \xnew \rangle \notag \\
            &= \edgeweight{rc} \frac{
                \sum\limits_{j = 1}^{n}{
                    \pi_j
                    (m_{cj} \cdot (1 - \delta\{j, c\}) - m_{rj} + 1)
                    \xnew_j
                }
            }{
                m_{rr} + \edgeweight{rc}(m_{rc} - m_{rr} + 1)
            }.
        \label{eq:fpi}
    \end{align}
    \label{thm:edge-score}
\end{theorem}
\begin{proof}
    The validity of the theorem's claim is straightforwardly established by
    substituting expression~(\ref{eq:single-edge-perturbation-pinew}) for $\pinew_j$
    from Theorem~\ref{thm:single-edge-perturbation} into
    $f_{\pi}(r, c) = \langle \pi - \pinew, \xnew \rangle$.
    \qedbull
\end{proof}
Essentially, Theorem~\ref{thm:edge-score} provides us with an edge score $f_{\pi}(r, c)$,
whose use for candidate edge selection comprises our heuristic for \diver. Unfortunately,
$f_{\pi}$'s computation is rather challenging, and is addressed in the
following section.

\subsection{Efficient Computation of Candidate Edge Scores}
\label{sec:fpi-efficient-comp}

Computation of candidate edge scores $f_\pi$ is challenging for two reasons. Firstly,
expression~(\ref{eq:fpi}) involves summation over all $n$ network nodes.
Since there are $\bigoh(n)$ candidate edges (with $n_{src} \ll n$ sources and $n$
destinations), it would result in at least a quadratic-time heuristic for \diver
that would not scale. Secondly,
expression~(\ref{eq:fpi}) involves mean first passage times, whose direct computation
is very expensive. We address both these challenges separately below.

\subsubsection{Focus on a Small Number of Nodes}
\label{sec:fpi-small-num-of-js}

Our first concern is that expression~(\ref{eq:fpi}) for $f_\pi$ contains summation over
all $n$ nodes. Intuitively, not all the nodes of a social network contribute equally to the value of~(\ref{eq:fpi}).
Indeed, in networks with skewed eigencentrality distribution, such as scale-free networks,
$f_\pi$ is mostly determined by a small number of top-centrality nodes. The latter is
illustrated in Fig.~\ref{fig:fpi-via-few-nodes}, where 50\% of the value of $f_\pi$ is
determined by less than 10\% of the top-centrality nodes of a scale-free
network\footnote{This would not be true for networks with ``uniform'' structure, such as
Erd\H{o}s–R\'{e}nyi (ER) or Watts-Strogatz (WS) networks. While we are not concerned with
such networks in this paper, it is still interesting to notice that, for ER and WS
networks, the value of $f_\pi$ is mostly determined by just two components of the
sum, corresponding to the candidate edge's ends.}.
\begin{figure}[ht]
    \centering
    \includegraphics[width=0.9\linewidth]{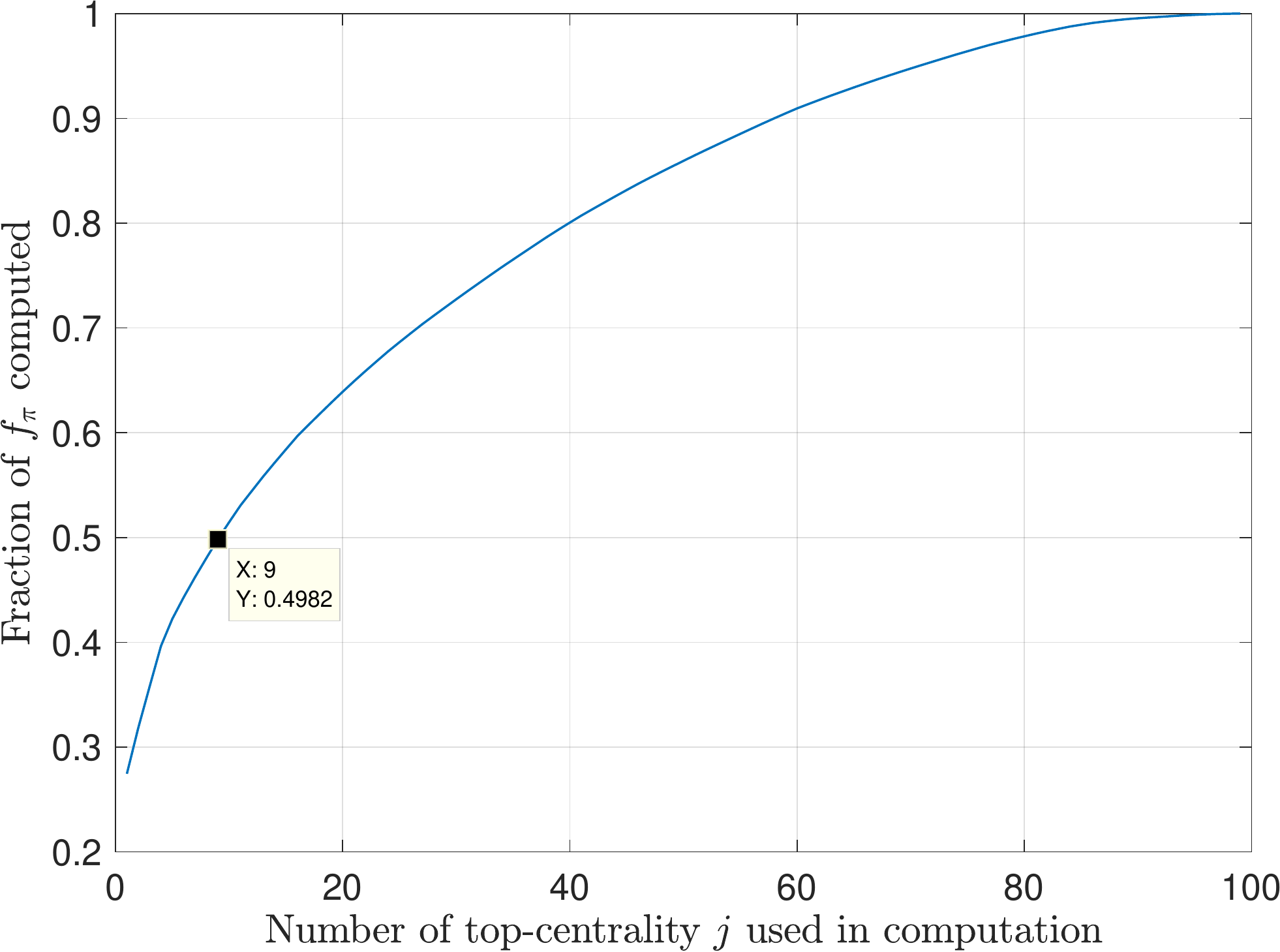}
    \caption{
        In a scale-free network $(n = 100, \gamma = -2.5)$, candidate edge scores
        $f_\pi$ are approximately computed using only a small fraction of network
        nodes in summation in~(\ref{eq:fpi}).
    }
    \label{fig:fpi-via-few-nodes}
\end{figure}

In Fig.~\ref{fig:fpi-exact-vs-approx}, we show how the approximately computed $f_\pi$
is related to its exactly
computed counterpart when we use different numbers of a scale-free network's nodes
for $f_\pi$'s computation. We can see that, even when we use only 10\% of nodes, the
relative order of $f_\pi$ for different candidate edges is close to the original, and
it is still easy to identify candidate edges $(r, c)$ with high values of $f_\pi(r ,c)$.
\begin{figure}[ht]
    \centering
    \includegraphics[width=0.9\linewidth]{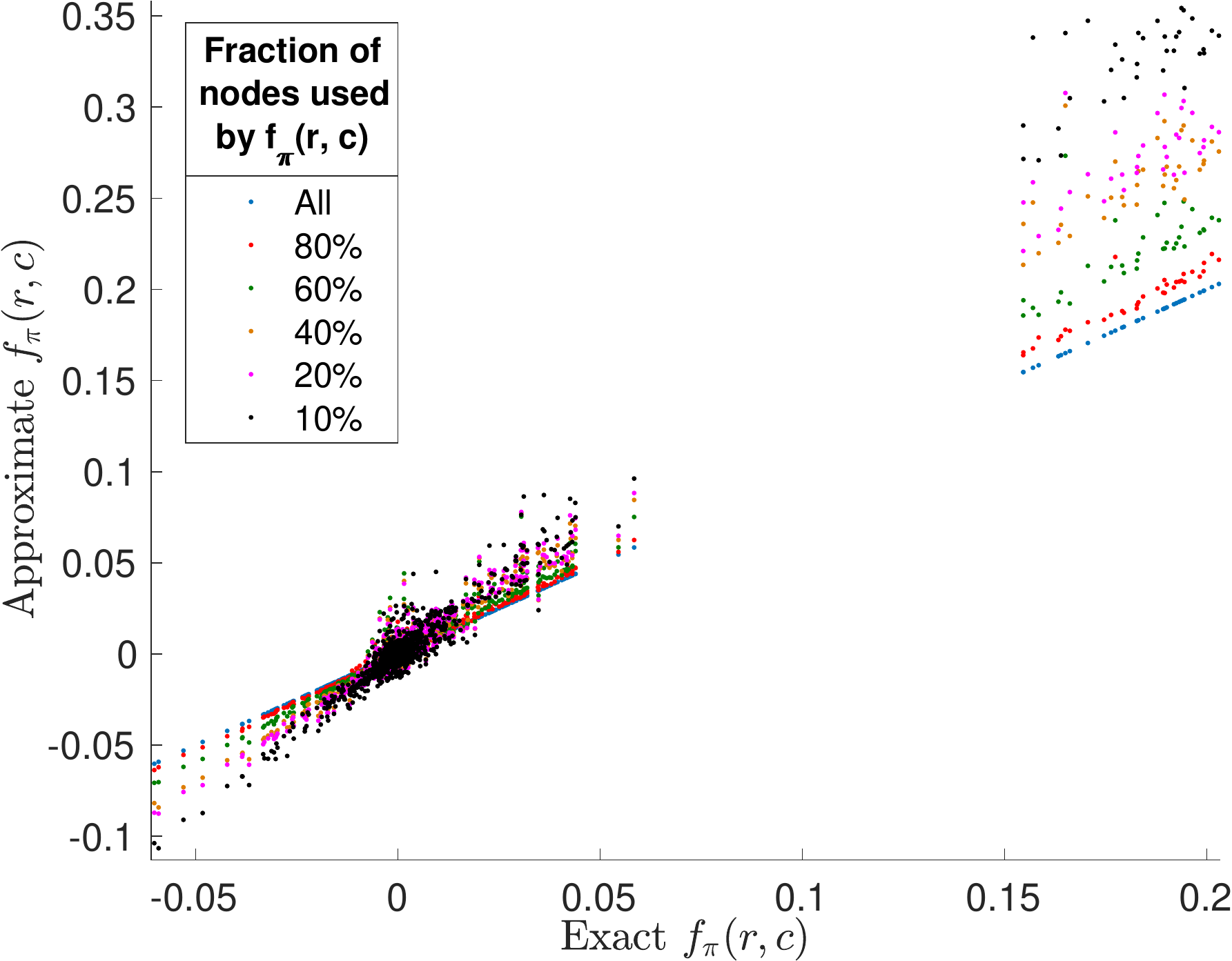}
    \caption{
        Comparison of exact and approximate candidate edge scores $f_\pi$
        in a scale-free network $(n = 100, \gamma = -2.5)$.
    }
    \label{fig:fpi-exact-vs-approx}
\end{figure}

Thus, to efficiently compute $f_\pi(r, c)$, we will use only those $j$ in~(\ref{eq:fpi})
corresponding to $n_{src}$ top-centrality nodes in the network
(in addition to $j = r, c$).

\subsubsection{Efficient Computation of Mean First Passage Times}
\label{sec:estimating-mfpt}

In the previous section, we have considerably simplified computation of
$f_\pi(r, c)$ by leaving only $\bigoh(n_{src})$ terms in expression~(\ref{eq:fpi}).
Now, our concern is to actually find values of the mean first passage times
$m_{ij}$ remaining in~(\ref{eq:fpi}).

The classic MFPT computation method of Kemeny and Snell~\cite{kemeny1976finite}
relies on the fundamental matrix $Z$ of Markov chain $W$, and defines MFPTs as
\begin{align*}
    Z &= (I - W + \onebb \pi^\trans)^{-1},\\
    M &= \{m_{ij}\} = (I - Z + \onebb \onebb^\trans \diag(Z)) \diag^{-1}(\onebb \pi).
\end{align*}
Computation of the fundamental matrix involves a cubic-time matrix inversion and
would not scale.
Hunter~\cite{hunter2017computation}
provides a survey of 11 alternative methods for MFPT computation, but all of them share
the same high complexity. Most importantly, however, all existing method target
computation of all $\bigoh(n^2)$ MFPTs between all the nodes in the network.

Let us notice that expression~(\ref{eq:fpi}) for $f_\pi$
\begin{align*}
    f_{\pi}(r, c) = \edgeweight{rc} \frac{
            \sum\limits_{j = 1}^{n}{
                \pi_j
                (m_{cj} \cdot (1 - \delta\{j, c\}) - m_{rj} + 1)
                \xnew_j
            }
        }{
            m_{rr} + \edgeweight{rc}(m_{rc} - m_{rr} + 1)
        }
    \tag{\ref{eq:fpi}}
\end{align*}
uses MFPTs either from or to high-centrality nodes: $r$ are top-centrality according
to Sec.~\ref{sec:edge-source-selection}; $j$ are top-centrality according to
Sec.~\ref{sec:fpi-small-num-of-js}). There are $n_{src} n \ll n^2$ such 
MFPTs, where $n_{src}$ is the number of candidate edge source nodes $r$.

We propose to estimate $n_{src} n$ MFPTs between a small number of nodes by
performing a finite-time random walk and estimating passage times between
the nodes. The walk starts at an arbitrary node, and proceeds for a predefined
number of hops following the transition probabilities defined by the adjacency
matrix $W$ viewed here as the state transition matrix of a Markov chain. While
performing the walk, we accumulate the passage times between $n_{src}$ candidate
edge sources and $n$ candidate edge destinations, and compute the means when
the walk is complete. This approach towards MFPT estimation is similar to the
$k$-Step Markov Approach that White and Smyth~\cite{white2003algorithms}
used for estimation of their MFPT-based relative importance of network nodes.

The key questions here are \emph{Will the proposed method result in good
estimates of MFPTs to and from high-centrality nodes?} and \emph{If so, how long
should that random walk be?} We answer these two questions via empirical analysis.

The first insight is that MFPTs to and from high-centrality nodes converge very fast,
since the walk visits such
nodes most often. This is illustrated in Fig.~\ref{fig:mfpt-est-vs-maxpi},
according to which the quality of MFPT estimates $m_{ij}$ noticeably varies with
the walk's length when both $i$ and $j$ are low-centrality, and is uniformly high
if at least one of $i$ and $j$ is a high-centrality node.
\begin{figure}[ht]
    \centering
    \includegraphics[width=0.85\linewidth]{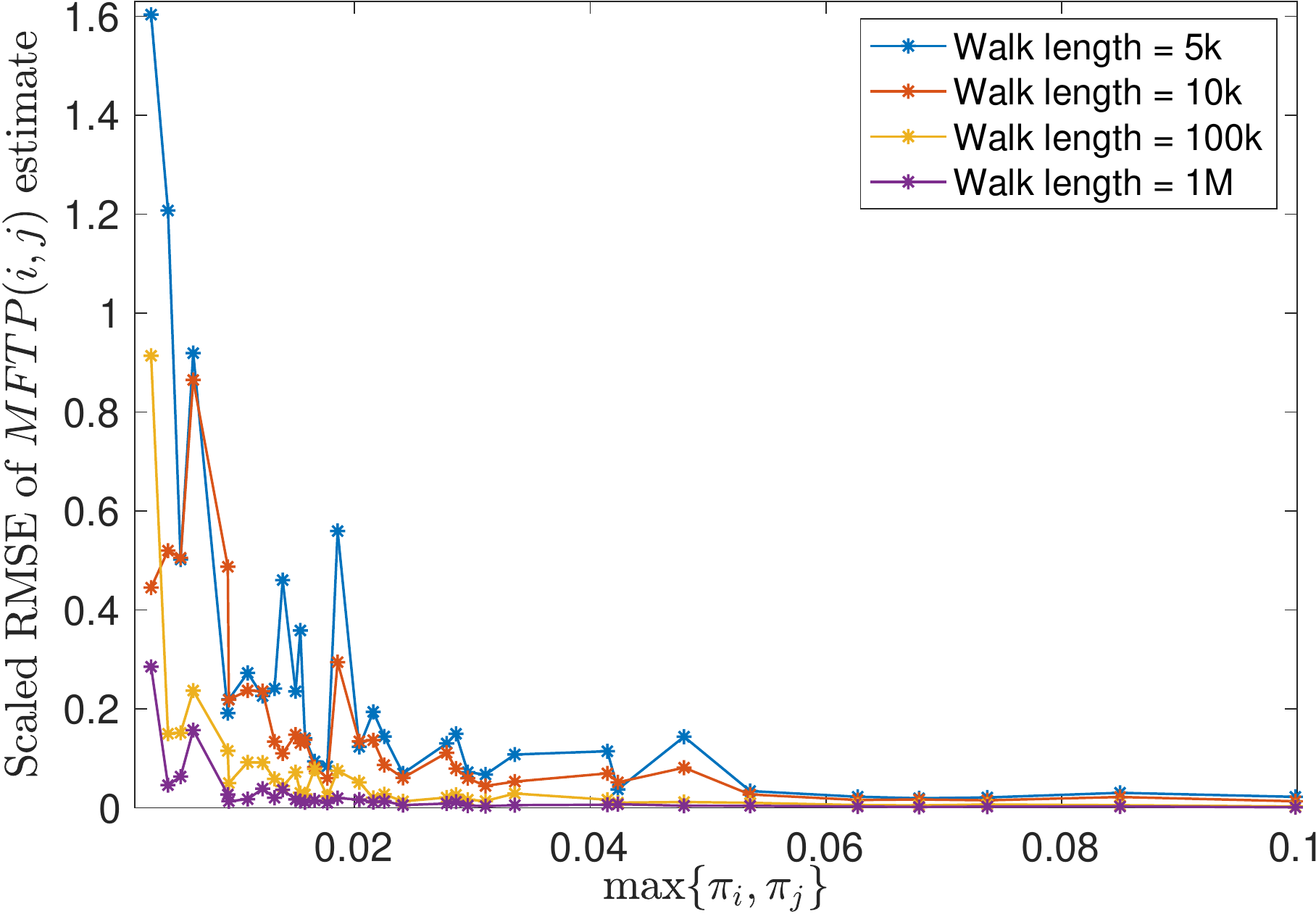}
    \caption{
        Quality of MFPT estimation in a scale-free network $(n=100, \gamma = -2.5)$
        using walks of different length.
    }
    \label{fig:mfpt-est-vs-maxpi}
\end{figure}
This insight echoes the result of Avrachenkov et al.~\cite{avrachenkov2007monte},
who show that PageRank estimates for high-centrality nodes obtained via Monte Carlo
simulation converge very fast.

Now, we empirically study the question of how long the random walk should be
to obtain sufficiently good estimates of MFPTs to and from top-centrality nodes
in a scale-free network, and report the results in Figures~\ref{fig:mfpt-est-rwlen}.
\begin{figure*}[ht]
    \centering
    \begin{subfigure}{0.48\linewidth}
        \includegraphics[width=\linewidth]{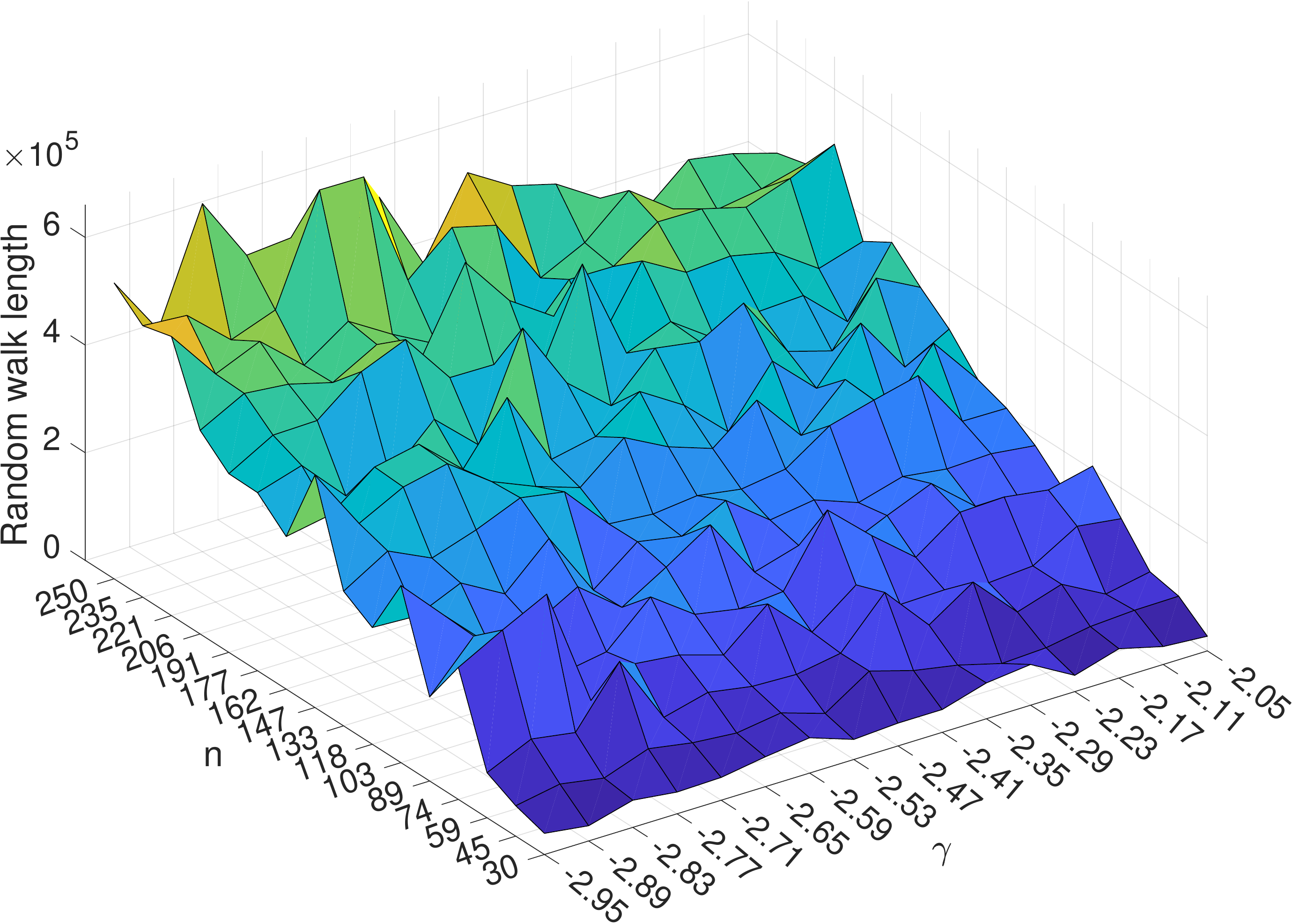}
        \caption{}
        \label{fig:mfpt-est-conv-niter-suff-3d}
    \end{subfigure}
    \hfill
    \begin{subfigure}{0.48\linewidth}
        \includegraphics[width=\linewidth]{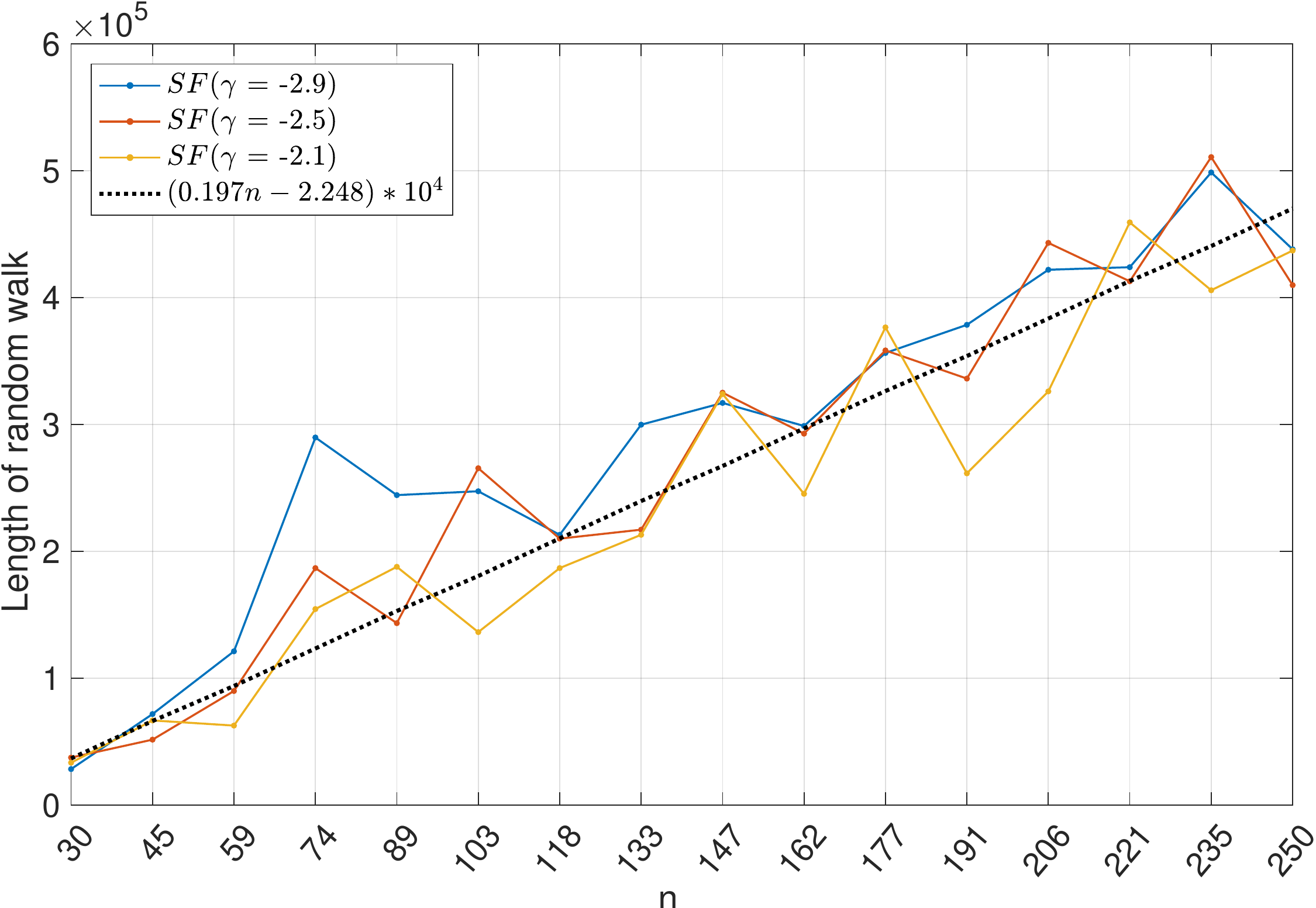}
        \caption{}
        \label{fig:mfpt-est-conv-niter-suff-heat-profile}
    \end{subfigure}
    \caption{
        Dependency of the length of a random walk---used for estimating MFPTs to and from
        top-centrality nodes---on the size and density of the scale-free network.
    }
    \label{fig:mfpt-est-rwlen}
\end{figure*}
%
Fig.~\ref{fig:mfpt-est-conv-niter-suff-3d} shows how many steps a random walk
should perform in order for 5\% of MFPTs to and from top 5\% high-centrality
nodes to converge within 5\% of their true values, while the network's
size $n$ and scale-free exponent $\gamma$ vary. For each pair $(n, \gamma)$,
100 networks are generated, and the mean walk lengths are
reported. Fig.~\ref{fig:mfpt-est-conv-niter-suff-heat-profile} shows the same
data for 3 specific scale-free exponents, $\gamma \in \{ -2.9, -2.5, -2.1\}$.
The length of the random walk does not depend on the scale-free exponent, and
depends upon the network's size $n$ as $(0.197n - 2.248) \cdot 10^4$.
These results allow to make the following statement.
\begin{proposition}[Random Walk Length]
    In scale-free networks, the length of a finite random walk sufficient
    for convergence of $\bigoh(n)$ MFPTs to and from $\bigoh(1)$ top-centrality
    nodes is $\bigoh(n)$ (in contrast to the $\bigoh(n^3)$ cost of the direct
    computation of all MFPTs via the fundamental matrix method).
    \label{thm:random-walk-length}
\end{proposition}

\subsection{Solving \texorpdfstring{\diver}{}}
\label{sec:solving-diver}

In this section, we gather all our results, formally state a heuristic for
solving \diver as Algorithm~\ref{alg:diver-heuristic}, and analyze its complexity
in Theorem~\ref{thm:diver-algo-time-complexity}.

\newcommand{\asep}[0]{\vspace{0.04in}}
\begin{algorithm}[ht] 
    \caption{Heuristic for \diver}
    \begin{algorithmic}[1]
        \renewcommand{\algorithmicrequire}{\textbf{Input:}}
        \renewcommand{\algorithmicensure}{\textbf{Output:}}

        \REQUIRE $W$---row-stochastic irreducible aperiodic sparse interpersonal appraisal matrix;
                 $k$---number of new edges to add;
                 $n_{src}$---maximal number of new edges' sources.
                 
        \ENSURE sequence $(r_1, c_1), (r_2, c_2), \dots$ of new edges to add
        
        \asep
        \STATE Compute eigenvector centrality $\pi$ \label{alg:diver-heuristic:step1}
                
        \asep
        \STATE Define candidate edge source nodes:  \label{alg:diver-heuristic:step2}\\
                $R \gets \text{$n_{src}$ top-centrality nodes in the network}$

        \asep
        \STATE{
            Estimate MFPTs $\{m_{ij}\}$ to and from each $r \in R$
            \label{alg:diver-heuristic:step4}
        }
        
        \asep
        \FOR{$r \in R$, $c \in \{1, \dots, n\}$ \label{alg:diver-heuristic:step5}}
            \STATE{
                Estimate $f_\pi(r, c)$ using $\bigoh(n_{src})$ top-centrality nodes
                \label{alg:diver-heuristic:step6}
            }
        \ENDFOR \label{alg:diver-heuristic:step7}
        
        \asep
        \STATE $S \gets \text{candidate edges $(r, c)$ having top-$k$ scores $f_\pi(r, c)$}$
            \label{alg:diver-heuristic:step8}
        
        \RETURN S
     \end{algorithmic}
\label{alg:diver-heuristic}
\end{algorithm}

\begin{theorem}
    Time-complexity of Algorithm~\ref{alg:diver-heuristic} is
    $\bigoh(n(\gap(W) + n_{src}^2) + n_{src} \log{n_{src}} + k \log{k})$,
    where $\gap(W)$ is the number of matrix-vector multiplications the power
    method uses to compute the dominant left eigenvector of $W$.
    \label{thm:diver-algo-time-complexity}
\end{theorem}
\begin{proof}
    In step~\ref{alg:diver-heuristic:step1} of Algorithm~\ref{alg:diver-heuristic},
    we compute the dominant left eigenvector $\pi$ of $W$, which can be done using
    power method. The later performs $\gap(W)$ matrix-vector multiplications, each
    of whom has a linear time complexity for sparse $W$. Thus, this step's complexity
    is $T_1 = \bigoh(\gap(W)n)$.
    The cost of selecting top $n_{src}$ elements out of $n$ at Step~\ref{alg:diver-heuristic:step2}
    is $T_2 = \bigoh(n + n_{src} \log(n_{src}))$.
    In step~\ref{alg:diver-heuristic:step4}, following Sec.~\ref{sec:estimating-mfpt} and,
    in particular, Proposition~\ref{thm:random-walk-length}, we estimate MFPTs via a
    $\bigoh(n)$-long finite random walk, so this step's cost is $T_3 = \bigoh(n)$.
    At steps~\ref{alg:diver-heuristic:step5}-\ref{alg:diver-heuristic:step7}, we compute
    $n_{src} n$ edge scores $f_\pi$. Following the method of Sec.~\ref{sec:fpi-small-num-of-js},
    each $f_{\pi}(r, c)$ is computed in time $\bigoh(n_{src})$, bringing time complexity
    of steps 4-6 to $T_{\ref{alg:diver-heuristic:step5}-\ref{alg:diver-heuristic:step7}} = \bigoh(n_{src}^2 n)$.
    Finally, selection of top $k$ out of $n_{src} n$ items at step~\ref{alg:diver-heuristic:step8}
    is performed in time $T_7 = \bigoh(n_{src} n + k \log{k})$.
    If we collect the expressions for $T_1, \dots, T_7$, we get
    $T = \bigoh(n(\gap(W) + n_{src}^2) + n_{src} \log{n_{src}} + k \log{k})$.
    \qedbull
\end{proof}

In Theorem~\ref{thm:diver-algo-time-complexity}, $\gap(W)$ is the number of iterations
required for convergence of the power method for computing eigencentrality vector $\pi$
of $W$. While the specific value of $\gap(W)$ depends on $W$'s spectral gap
$\lambda_2 / \lambda_1$, in practise, $\gap(W)$ is usually assumed to be a reasonably
small constant. Thus, assuming that $\gap(W)$ is bounded, as well as noticing that
we choose both $n_{src}$ and $k$ to be small, that is, $n_{src} \ll n$ and $k \ll n$,
it immediately follows from Theorem~\ref{thm:diver-algo-time-complexity} that
Algorithm~\ref{alg:diver-heuristic} is computable in time $\bigoh(n)$. For practical
purposes, the hidden constant factor can be reduced by considering only some of $n$
nodes as destinations for the candidate edges; for example, we can consider only the
destinations 2 hops away from the sources, increasing the acceptance likelihood of the
recommended edges.

\section{Experimental Results}
\label{sec:experiments}

We use Algorithm~\ref{alg:diver-heuristic} to solve \diver on scale-free networks
$(n = 250, \gamma = -2.5)$. The initial user opinions $x \in [0, 1]^n$ are generated
uniformly at random.
The adversary uniformly randomly selects 16 users and changes their opinions to $1.0$,
defining $\xnew$. At this stage, our goal is to add edges to the network to make the new
asymptotic consensus value \acvnewxpis as close as possible to the original \acvs.
We start adding new edges to the network, 5 at a time, using
Algorithm~\ref{alg:diver-heuristic} for edge selection. The number of candidate edge
source nodes $n_{src}$ is $25$. The edge addition process stops either when the number
of added edges reaches 180, or if \diver's signed objective $\acvnewxpis - \acvs$ has
become smaller than $10^{-8}$. The results are reported in Figures~\ref{fig:exp-results}.
\begin{figure*}[ht]
    \centering
    \begin{subfigure}{0.48\linewidth}
        \includegraphics[width=\linewidth]{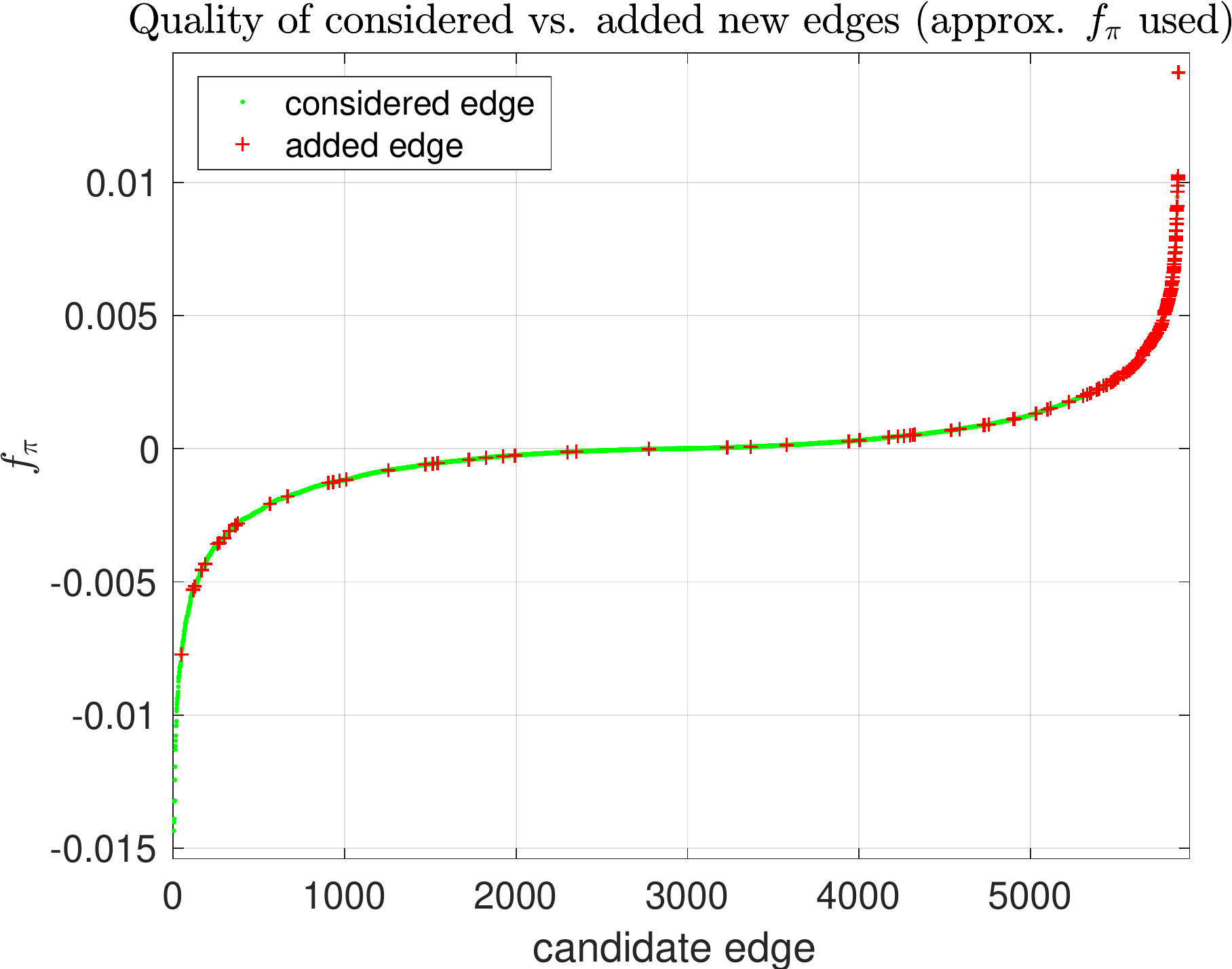}
        \caption{}
        \label{fig:edge-selection}
    \end{subfigure}
    \hfill
    \begin{subfigure}{0.48\linewidth}
        \includegraphics[width=\linewidth]{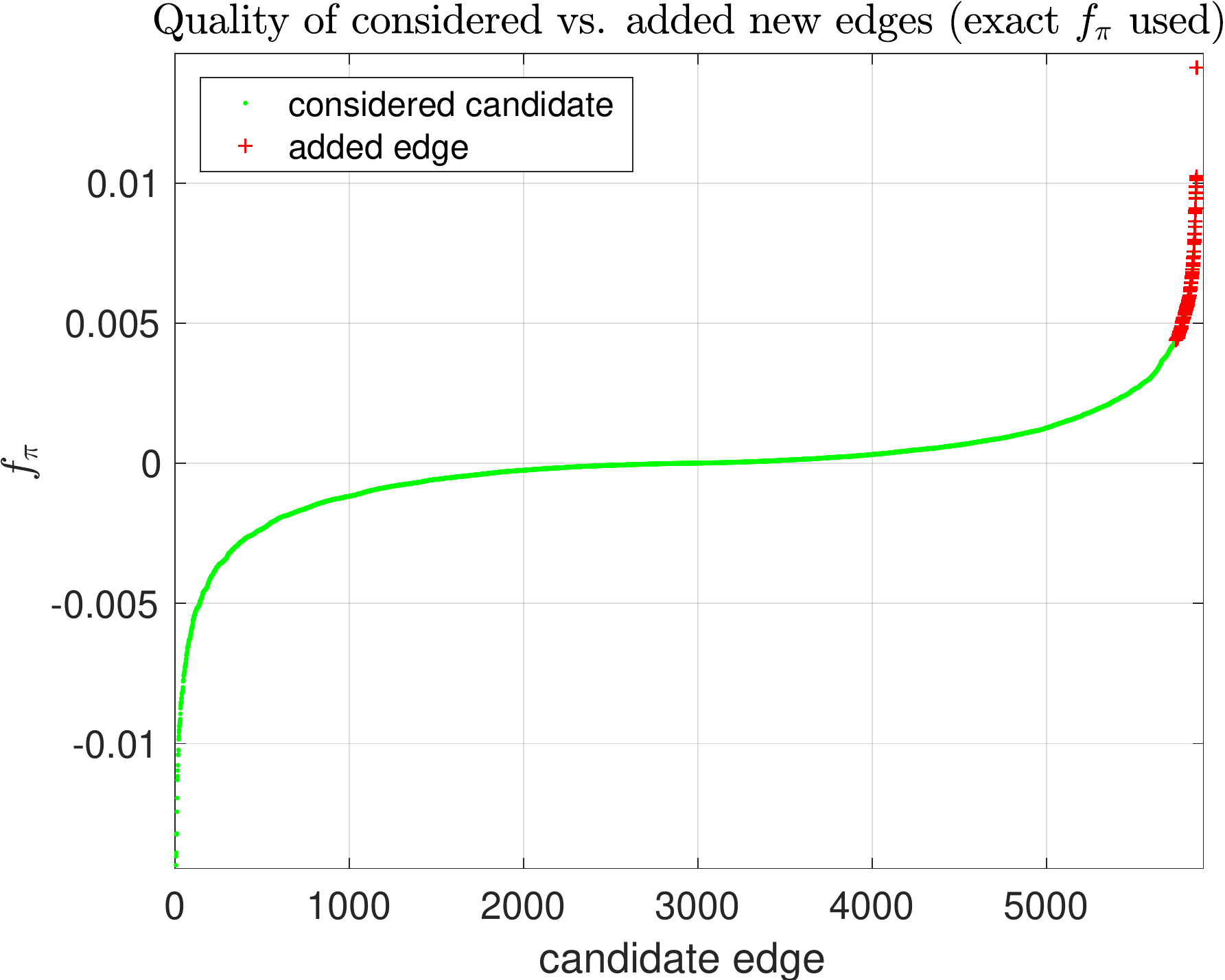}
        \caption{}
        \label{fig:edge-selection-exact-fpi}
    \end{subfigure}
    \newline
    \begin{subfigure}{0.48\linewidth}
        \includegraphics[width=\linewidth]{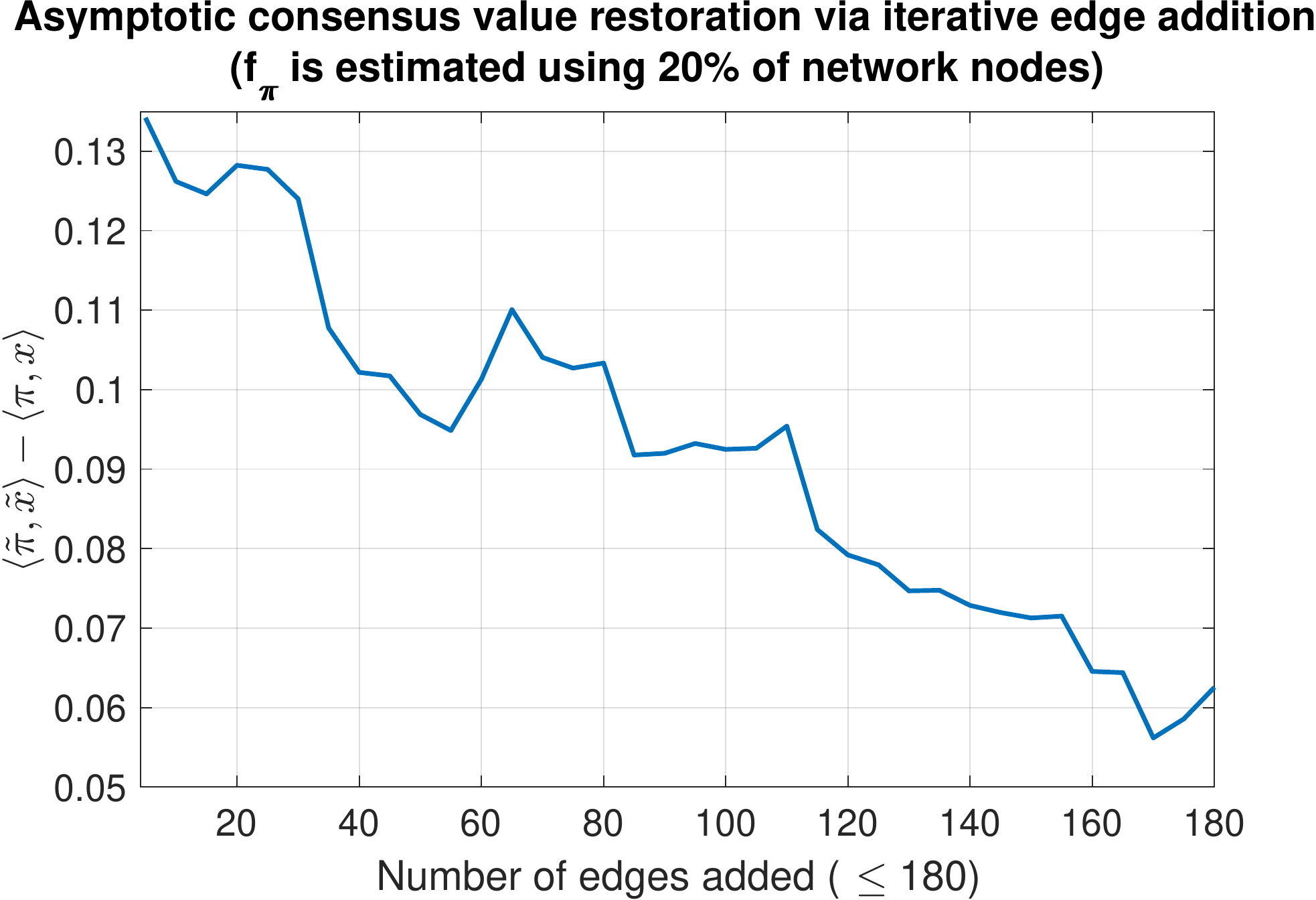}
        \caption{}
        \label{fig:diver-run}
    \end{subfigure}
    \hfill
    \begin{subfigure}{0.48\linewidth}
        \includegraphics[width=\linewidth]{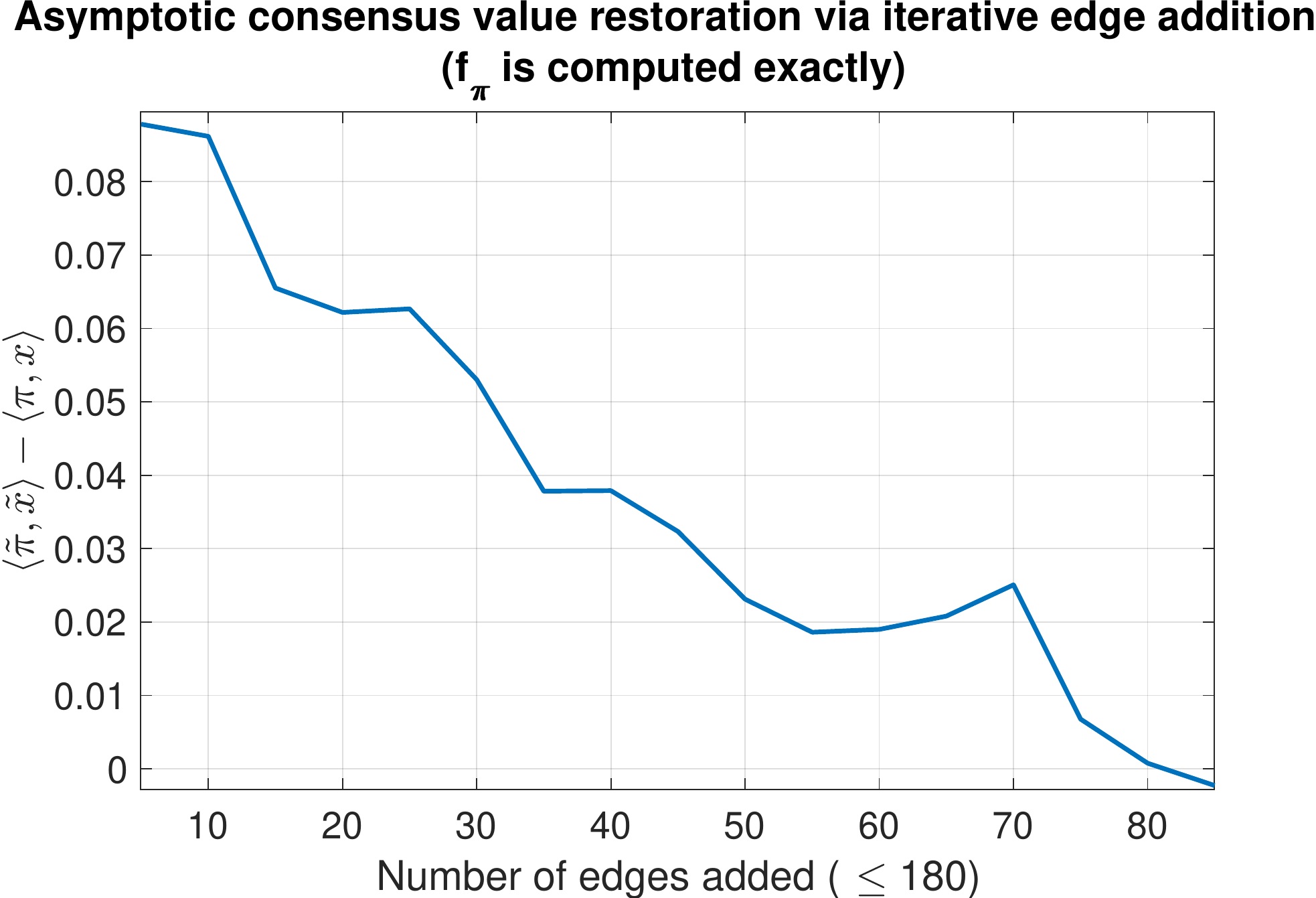}
        \caption{}
        \label{fig:diver-run-exact-fpi}
    \end{subfigure}
    \caption{
        Solving \diver via iterative edge addition using Algorithm~\ref{alg:diver-heuristic}.
    }
    \label{fig:exp-results}
\end{figure*}

Figures~\ref{fig:edge-selection} and~\ref{fig:edge-selection-exact-fpi} show the quality $f_\pi$ of the edges that
were suggested by Algorithm~\ref{alg:diver-heuristic} for addition to the network, where
the edge quality is either estimated or computed exactly, respectively.
In Fig.~\ref{fig:edge-selection},
we can see that, if $f_{\pi}$ are estimated using only 20\% of top-centrality nodes in the
computation, then a few ``bad'' edges are selected by the heuristic, but most of the added
edges are top-quality (have large values of $f_{\pi}$). In Fig.~\ref{fig:edge-selection-exact-fpi},
for the case of exact computation of $f_\pi$, all added edges are top-quality, though, a small
number of top-quality edges are missed, because their source nodes are not among the $n_{src}$ ones considered.

Fig.~\ref{fig:diver-run} and Fig.~\ref{fig:diver-run-exact-fpi} show how \diver's objective
changes during iterative edge addition, when $f_\pi$ are either estimated or computed exactly.
When we use exact $f_\pi$ (Fig.~\ref{fig:diver-run-exact-fpi}), it takes 17 iterations
(85 new edges) to drive the asymptotic consensus value \acvnewxpis back close to its original
state \acvs. In case of using estimates of $f_\pi$ (Fig.~\ref{fig:diver-run}), the process
stops having reached the allowed maximum of 180 new edges, and the asymptotic consensus
value \acvnewxpis gets within $\approx 0.06$ (10\%) of its original state \acvs.

\section{Discussion and Future Work}
\label{sec:discussion-and-future}

In this work, we have formulated \diver---a problem of strategically adding edges to the
network in order to disable the effect of external influence altering opinions of select
users. Using a functional reduction from the classic subset sum problem, we have proven
that this problem is NP-hard even for the case of undirected networks. Due to the problem's
hardness, we have focused on designing a heuristic for it. To that end, we have provided
a perturbation analysis, formally answering the question of how the network nodes'
eigencentralities and, hence, \diver's objective function change when a single edge is
added to the network. The latter analysis led to the definition of candidate edge scores,
that quantify the potential impact of the candidate edges, allowing to add them to the network
in a greedy fashion. We have also provided insights into how to compute these edge scores
in scale-free-like networks in pseudo-constant time, resulting in a pseudo-linear-time
heuristic for \diver. One of these insights is related to efficiently estimating mean first
passage times in Markov chains to and from high-centrality states. We have confirmed our
theoretical findings in experiments.

Our results are rather general, and the provided insights and theory can be applied to
other problems of strategically manipulating eigenvector centrality in networks. Simultaneously,
this work opens many avenues for potential \emph{future research}, including the following.

\vspace{0.04in} \noindent \hspace{0.04in} \bull
\diver targets optimization of the scalar asymptotic consensus value, 
being the value to which all the opinions asymptotically converge under DeGroot model.
It is possible to generalize \diver to the case of such opinion dynamics models as Friedkin-Johnsen~\cite{friedkin1999social} or the non-linear DeGroot
model~\cite{amelkin2017polar}, where the users asymptotically disagree, and, hence,
we need to optimize the asymptotic opinion distribution, rather than a scalar.

\vspace{0.04in} \noindent \hspace{0.04in} \bull
For the purposes of efficiently computing potential impact of candidate edge addition,
we have studied the question of how to efficiently estimate mean first passage times to
or from top-centrality states in a Markov chain. While our answer to the latter question
was based on an empirical study on scale-free networks, providing formal convergence
bounds for MFPT estimation would benefit many scientific areas dealing with Markov
processes.

\vspace{0.04in} \noindent \hspace{0.04in} \bull
Finally, while \diver was proposed as a method for fighting external influence upon
the opinions of a social network's users, it clearly can be used as an influence
tool. Thus, it would be fruitful to study the ways to identify whether an edge
recommendation process in a social network targets strategic change of the opinion
distribution.

\bibliographystyle{ieeetr}
\bibliography{biblio-alias,biblio}

\end{document}